\documentclass{lmcs}
\usepackage{lastpage}
\lmcsdoi{18}{2}{6}
\lmcsheading{}{\pageref{LastPage}}{}{}%
{Dec.~01,~2020}{May~05,~2022}{}

\usepackage{cite}
\usepackage{color}
\usepackage{euscript}
\usepackage{graphicx}
\usepackage{mathrsfs} 
\usepackage{microtype}
\usepackage[normalem]{ulem}
\usepackage{wrapfig}

\def\ttt{\texttt}

\def\extraspacing{\vspace{3mm} \noindent}
\def\figcapup{\vspace{-1mm}}
\def\figcapdown{\vspace{-0mm}}

\def\vgap{\vspace{1mm}}

\newcommand{\bm}[1]{\textrm{\boldmath${#1}$}}

\newcommand{\myeqn}[1]{\begin{eqnarray}#1\end{eqnarray}}
\newcommand{\myset}[1]{\{#1\}}
\newcommand{\set}[1]{\{#1\}}

\newcommand{\explain}[1]{(\textrm{#1})}

\def\mit{\mathit}

\def\eps{\epsilon}
\def\fr{\frac}
\def\-{\mbox{-}}

\def\tO{\tilde{O}}

\def\lc{\lceil}
\def\lf{\lfloor}
\def\rc{\rceil}
\def\rf{\rfloor}

\def\nn{\nonumber}

\def\*{\star}

\def\done{\qed \vspace{5mm}}	

\usepackage{autonum}
\usepackage[utf8]{inputenc}

\def\figcapup{\vspace{0mm}}
\def\figcapdown{\vspace{0mm}}
\def\extraspacing{\vspace{5mm} \noindent}
\def\vgap{\vspace{3mm}}
\def\A{\mathcal{A}}
\def\E{\mathcal{E}}

\def\G{\mathcal{G}}
\def\H{\mathcal{H}}
\def\I{\mathcal{I}}
\def\J{\mathcal{J}}
\def\L{\mathcal{L}}
\def\Q{\mathcal{Q}}

\def\U{\mathcal{U}}
\def\V{\mathcal{V}}
\def\Z{\mathcal{Z}}

\def\tO{\tilde{O}}
\def\attset{\mathit{attset}}
\def\att{\textrm{\bf att}}
\def\config{\mathit{config}}
\def\cp{\mathit{Join}}
\def\dom{\textrm{\bf dom}}
\def\iso{\mathit{isolated}}
\def\join{\mathit{Join}}
\def\scheme{\mathit{scheme}}
\def\share{\mathit{share}}

\begin{document}

\title{A Near-Optimal Parallel Algorithm\texorpdfstring{\\}{ }for Joining Binary Relations}

\author[B.~Ketsman]{Bas Ketsman}
\address{Vrije Universiteit Brussel}	
\email{\texttt{bas.ketsman@vub.be}} 

\author[D.~Suciu]{Dan Suciu}
\address{University of Washington}	
\email{\texttt{suciu@cs.washington.edu}}

\author[Y.~Tao]{Yufei Tao}
\address{Chinese University of Hong Kong, Hong Kong}	
\email{\texttt{taoyf@cse.cuhk.edu.hk}}
\thanks{The research of Bas Ketsman was partially supported by FWO-grant G062721N. The research of Dan Suciu was partially supported by projects NSF IIS 1907997 and NSF-BSF 2109922. The research of Yufei Tao was partially supported by GRF projects 142034/21 and 142078/20 from HKRGC, and an AIR project from the Alibaba group.}


\renewcommand*\L{\mathcal{L}}
\begin{abstract}
    We present a constant-round algorithm in the massively parallel computation (MPC) model for evaluating a natural join where every input relation has two attributes. Our algorithm achieves a load of $\tO(m/p^{1/\rho})$ where $m$ is the total size of the input relations, $p$ is the number of machines, $\rho$ is the join's fractional edge covering number, and $\tO(.)$ hides a polylogarithmic factor. The load matches a known lower bound up to a polylogarithmic factor. At the core of the proposed algorithm is a new theorem (which we name {\em the isolated cartesian product theorem}) that provides fresh insight into the problem's mathematical structure. Our result implies that the {\em subgraph enumeration problem}, where the goal is to report all the occurrences of a constant-sized subgraph pattern, can be settled optimally (up to a polylogarithmic factor) in the MPC model.
\end{abstract}

\maketitle

\section{Introduction} \label{sec&intro} 

Understanding the hardness of joins has been a central topic in database theory. Traditional efforts have focused on discovering fast algorithms for processing joins in the {\em random access machine} (RAM) model (see \cite{agm13,ahv95,nprr18,nrr13,ps14,v14,y81} and the references therein). Nowadays, massively parallel systems such as Hadoop \cite{dg04} and Spark \cite{aba+09} have become the mainstream architecture for analytical tasks on gigantic volumes of data. Direct adaptations of RAM algorithms, which are designed to reduce CPU time, rarely give satisfactory performance on that architecture. In systems like Hadoop and Spark, it is crucial to minimize communication across the participating machines because usually the overhead of message exchanging overwhelms the CPU calculation cost. This has motivated a line of research --- which includes this work --- that aims to understand the communication complexities of join problems. 

\subsection{Problem Definition} \label{sec&intro-prob}

We will first give a formal definition of the join operation studied in this paper and then elaborate on the computation model assumed.

\extraspacing {\bf Joins.} Let $\att$ be a finite set where each element is called an {\em attribute}, and $\dom$ be a countably infinite set where each element is called a {\em value}. A {\em tuple} over a set $U \subseteq \att$ is a function $\bm{u}: U \rightarrow \dom$. Given a subset $V$ of $U$, define $\bm{u}[V]$ as the tuple $\bm{v}$ over $V$ such that $\bm{v}(X) = \bm{u}(X)$ for every $X \in V$. We say that $\bm{u}[V]$ is the {\em projection} of $\bm{u}$ on $V$.

\vgap

A {\em relation} is a set $R$ of tuples over the same set $U$ of attributes. We say that the {\em scheme} of $R$ is $U$, and write this fact as $\scheme(R) = U$. $R$ is {\em unary} or {\em binary} if $|\scheme(R)| = 1$ or $2$, respectively. A value $x \in \dom$ {\em appears} in $R$ if there exist a tuple $\bm{u} \in R$ and an attribute $X \in U$ such that $\bm{u}(X) = x$; we will also use the expression that $x$ is ``a value on the attribute $X$ in $R$''.

\vgap 

A {\em join query} (sometimes abbreviated as a ``join'' or a ``query'') is a set $\Q$ of relations. Define  
$\attset(\Q) = \bigcup_{R \in \Q} \scheme(R)$.
The {\em result} of the query, denoted as $\join(\Q)$, is the following relation over $\attset(\Q)$
\myeqn{
    \Big\{ 
    \textrm{tuple $\bm{u}$ over $\attset(\Q)$} \bigm| \forall R \in \Q, \, \textrm{$\bm{u}[\scheme(R)] \in R$} 
    \Big\}. 
    \nn 
}
$\Q$ is 
\begin{itemize} 
    \item {\em simple} if no distinct $R, S \in \Q$ satisfy $\scheme(R) = \scheme(S)$;  
    \item {\em binary} if every $R \in \Q$ is binary. 
\end{itemize}
Our objective is to design algorithms for answering simple binary queries. 

\vgap

The integer 
\myeqn{
    m=
    \sum_{R \in \Q} |R|
    \label{eqn&inputsize}
}
is the {\em input size} of $\Q$. Concentrating on {\em data complexity}, we will assume that both $|\Q|$ and $|\attset(\Q)|$ are constants.

 \extraspacing {\bf Computation Model.} We will assume the {\em massively parallel computation} (MPC) model, which is a widely-accepted abstraction of today's massively parallel systems. Denote by $p$ the number of machines. In the beginning, the input elements are evenly distributed across these machines. For a join query, this means that each machine stores $\Theta(m/p)$ tuples from the input relations (we consider that every value in $\dom$ can be encoded in a single word). 

\vgap

An algorithm is executed in {\em rounds}, each having two phases: 
\begin{itemize} 
    \item In the first phase, every machine performs computation on the data of its local storage. 
   
    \item In the second phase, the machines communicate by sending messages to each other. 
\end{itemize}
All the messages sent out in the second phase must be prepared in the first phase. This prevents a machine from, for example, sending information based on the data received {\em during} the second phase. Another round is launched only if the current round has not solved the problem. In our context, {\em solving} a join query means that, for every tuple $\bm{u}$ in the join result, at least one of the machines has $\bm{u}$ in the local storage; furthermore, no tuples outside the join result should be produced. 

\vgap

The {\em load} of a round is the largest number of words received by a machine in this round, that is, if machine $i \in [1, p]$ receives $x_i$ words, the load is $\max_{i=1}^p x_i$. The performance of an algorithm is measured by two metrics: (i) the number of rounds, and (ii) the {\em load} of the algorithm, defined as the total load of all rounds. CPU computation is for free. We will be interested only in algorithms finishing in a constant number of rounds. The load of such an algorithm is asymptotically the same as the maximum load of the individual rounds.

\vgap

The number $p$ of machines is assumed to be significantly less than $m$, which in this paper means $p^3 \le m$. For a randomized algorithm, when we say that its load is at most $L$, we mean that its load is bounded by $L$ with probability at least $1-1/p^c$ where $c$ can be set to an arbitrarily large constant. The notation $\tO(.)$ hides a factor that is polylogarithmic to $m$ and $p$.

\subsection{Previous Results} \label{sec&intro-prev}

Early work on join processing in the MPC model aimed to design algorithms performing only one round. Afrati and Ullman \cite{au11} explained how to answer a query $\Q$ with load $O(m/p^{1/|\Q|})$. Later, by refining their prior work in \cite{bks17c}, Koutris, Beame, and Suciu \cite{kbs16} described an algorithm that can guarantee a load of $\tO(m/p^{1/\psi})$, where $\psi$ is the query's {\em fractional edge quasi-packing number}. To follow our discussion in Section~\ref{sec&intro}, the reader does not need the formal definition of $\psi$ (which will be given in Section~\ref{sec&hypergraph}); it suffices to understand that $\psi$ is a positive constant which can vary significantly depending on $\Q$. In \cite{kbs16}, the authors also proved that any one-round algorithm must incur a load of $\Omega(m / p^{1/\psi})$, under certain assumptions on the statistics available to the algorithm.

\vgap 

Departing from the one-round restriction, subsequent research has focused on algorithms performing multiple, albeit still a constant number of, rounds. The community already knows \cite{kbs16} that any constant-round algorithm must incur a load of $\Omega (m/p^{1/\rho})$ answering a query, where $\rho$ is the query's {\em fractional edge covering number}. As far as Section~\ref{sec&intro} is concerned, the reader does not need to worry about the definition of $\rho$ (which will appear in Section~\ref{sec&hypergraph}); it suffices to remember two facts: 
\begin{itemize}
    \item Like $\psi$, $\rho$ is a positive constant which can vary significantly depending on the query $\Q$. 
    
    \item On the same $\Q$, $\rho$ never exceeds $\psi$, but can be much smaller than $\psi$ (more details in Section~\ref{sec&hypergraph}). 
\end{itemize}
The second bullet indicates that $m/p^{1/\rho}$ can be far less than $m/p^{1/\psi}$, suggesting that we may hope to significantly reduce the load by going beyond only one round. Matching the lower bound $\Omega(m/p^{1/\rho})$ with a concrete algorithm has been shown possible for several special query classes, including star joins \cite{au11}, cycle joins \cite{kbs16}, clique joins \cite{kbs16}, line joins \cite{au11,kbs16}, Loomis-Whitney joins \cite{kbs16}, etc. The simple binary join defined in Section~\ref{sec&intro-prob} captures cycle, clique, and line joins as special cases. Guaranteeing a load of $O(m/p^{1/\rho})$ for arbitrary simple binary queries is still open. 

\subsection{Our Contributions} \label{sec&intro-ours}
The paper's main algorithmic contribution is to settle any simple binary join $\Q$ under the MPC model with load $\tO(m / p^{1/\rho})$ in a constant number rounds (Theorem~\ref{thm&alg-main}). The load is optimal up to a polylogarithmic factor. Our algorithm owes to a new theorem ---  we name the {\em isolated cartesian product theorem} (Theorem~\ref{thm&cpthm}; see also Theorem~\ref{thm&cpthm2}) --- that reveals a non-trivial fact on the problem's mathematical structure.  

\extraspacing {\bf Overview of Our Techniques.} Consider the join query $\Q$ illustrated by the graph in Figure~\ref{fig&intro-join}a. An edge connecting vertices $X$ and $Y$ represents a relation $R_{\{X,Y\}}$ with scheme $\{X, Y\}$. $\Q$ contains all the 18 relations represented by the edges in Figure~\ref{fig&intro-join}a; $\attset(\Q) = \{\ttt{A}, \ttt{B}, ..., \ttt{L}\}$ has a size of 12. 

\vgap

Set $\lambda = \Theta(p^{1/(2\rho)})$ where $\rho$ is the fractional edge covering number of $\Q$ (Section~\ref{sec&hypergraph}). A value $x \in \dom$ is {\em heavy} if at least $m/\lambda$ tuples in an input relation $R \in \Q$ carry $x$ on the same attribute. The number of heavy values is $O(\lambda)$. A value $x \in \dom$ is {\em light} if $x$ appears in at least one relation $R \in \Q$ but is not heavy. A tuple in the join result may take a heavy or light value on each of the 12 attributes $\ttt{A}, ..., \ttt{L}$. As there are $O(\lambda)$ choices on each attribute (i.e., either a light value or one of the $O(\lambda)$ heavy values), there are $t = O(\lambda^{12})$ ``choice combinations'' from all attributes; we will refer to each combination as a {\em configuration}. Our plan is to partition the set of $p$ servers into $t$ subsets of sizes $p_1, p_2, ..., p_t$ with $\sum_{i=1}^t p_i = p$, and then dedicate $p_i$ servers ($1 \le i \le t$) to computing the result tuples of the $i$-th configuration. This can be done in parallel for all $O(\lambda^{12})$ configurations. The challenge is to compute the query on each configuration with a load $O(m/p^{1/\rho})$, given that only $p_i$ (which can be far less than $p$) servers are available for that subtask.

\begin{figure} 
    \centering 
    \begin{tabular}{cccc}
        \includegraphics[height=30mm]{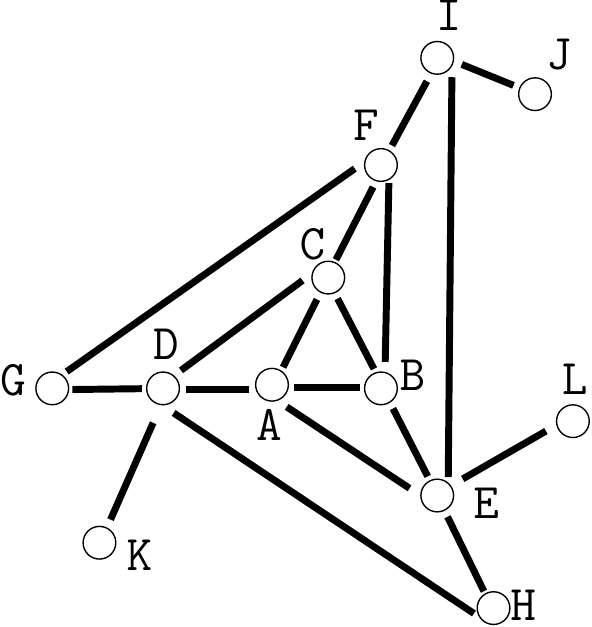}
        &
        \includegraphics[height=30mm]{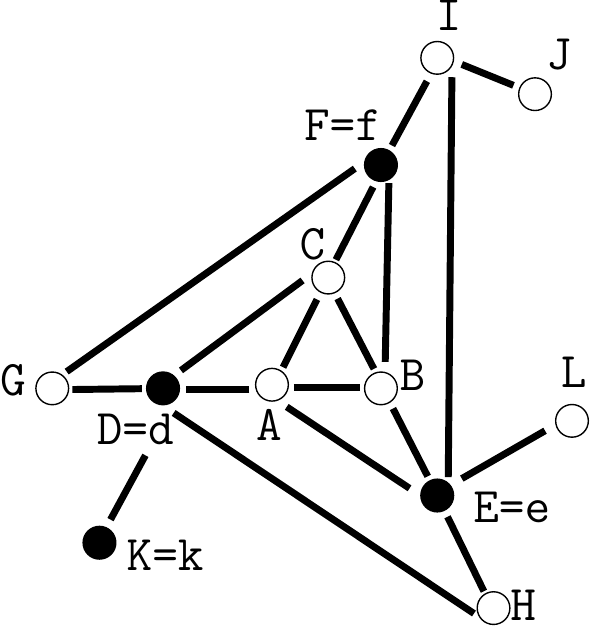}
        &
        \includegraphics[height=30mm]{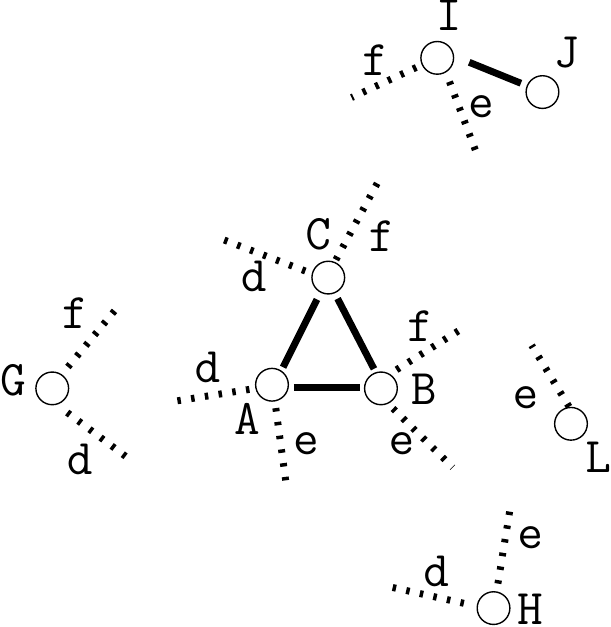} 
        &
        \includegraphics[height=30mm]{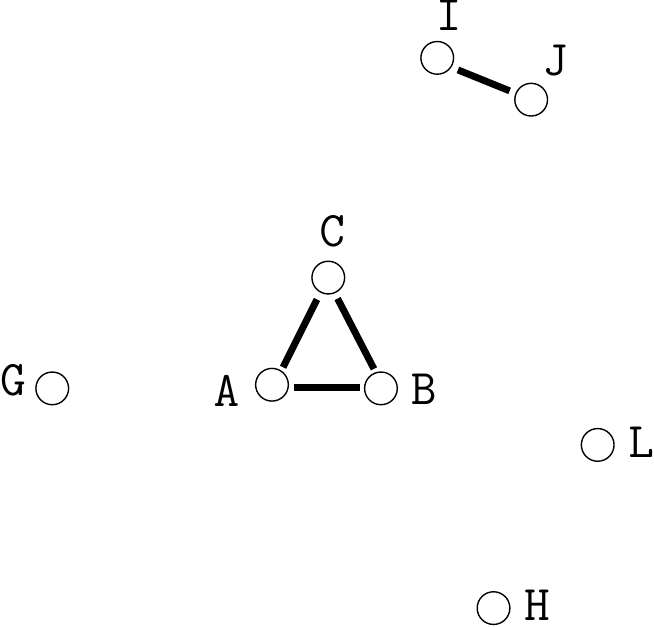} 
        \\
        (a) A join query
        &
        (b) A residual query 
        &
        (c) After deleting  
        &
        (d) After semi-join \\[-1mm]
        &
        &
        black verticess
        &
        reduction
    \end{tabular}
    \figcapup 
    \caption{Processing a join by constraining heavy values} 
    \label{fig&intro-join}
    \figcapdown 
\end{figure}

\vgap

Figure~\ref{fig&intro-join}b illustrates one possible configuration where we constrain attributes $\ttt{D}$, $\ttt{E}$, $\ttt{F}$, and $\ttt{K}$ respectively to heavy values $\ttt{d}$, $\ttt{e}$, $\ttt{f}$, and $\ttt{k}$ and the other attributes to light values. Accordingly, vertices $\ttt{D}$, $\ttt{E}$, $\ttt{F}$, and $\ttt{K}$ are colored black in the figure. The configuration gives rise to a {\em residual query} $\Q'$: 
\begin{itemize}
    \item For each edge $\{X,Y\}$ with two white vertices, $\Q'$ has a relation $R'_{\{X,Y\}}$ that contains only the tuples in $R_{\{X,Y\}} \in \Q$ using light values on both $X$ and $Y$; 
    
    \item For each edge $\{X,Y\}$ with a white vertex $X$ and a black vertex $Y$, $\Q'$ has a relation $R'_{\{X,Y\}}$ that contains only the tuples in $R_{\{X,Y\}} \in \Q$ each using a light value on $X$ and the constrained heavy value on $Y$;  
    
    \item For each edge $\{X,Y\}$ with two black vertices, $\Q'$ has a  relation $R'_{\{X,Y\}}$ with only one tuple that takes the constrained heavy values on $X$ and $Y$, respectively. 
\end{itemize}
For example, a tuple in $R'_{\{\ttt{A},\ttt{B}\}}$ must use light values on both $\ttt{A}$ and $\ttt{B}$; a tuple in $R'_{\{\ttt{D},\ttt{G}\}}$ must use value $\ttt{d}$ on $\ttt{D}$ and a light value on $\ttt{G}$; $R'_{\{\ttt{D},\ttt{K}\}}$ has only a single tuple with values $\ttt{d}$ and $\ttt{k}$ on $\ttt{D}$ and $\ttt{K}$, respectively. Finding all result tuples for $\Q$ under the designated configuration amounts to evaluating $\Q'$.

\vgap 

Since the black attributes have had their values fixed in the configuration, they can be deleted from the residual query, after which some relations in $\Q'$ become unary or even disappear. Relation $R'_{\{\ttt{A},\ttt{D}\}} \in \Q'$, for example, can be regarded as a unary relation over $\{\ttt{A}\}$ where every tuple is ``piggybacked'' the value $\ttt{d}$ on $\ttt{D}$. Let us denote this unary relation as $R'_{\{\ttt{A}\}|\ttt{d}}$, which is illustrated 
in Figure~\ref{fig&intro-join}c with a dotted edge extending from $\ttt{A}$ and carrying the label $\ttt{d}$. The deletion of $\ttt{D}, \ttt{E}$, $\ttt{F}$, and $\ttt{K}$ results in 13 unary relations (e.g., two of them are over $\{\ttt{A}\}$: $R'_{\{\ttt{A}\}|\ttt{d}}$ and $R'_{\{\ttt{A}\}|\ttt{e}}$). Attributes \ttt{G}, \ttt{H}, and \ttt{L} become {\em isolated} because they are not connected to any other vertices by solid edges. Relations $R'_{\{\ttt{A}, \ttt{B}\}}$, $R'_{\{\ttt{A}, \ttt{C}\}}$, $R'_{\{\ttt{B}, \ttt{C}\}}$, and $R'_{\{\ttt{I}, \ttt{J}\}}$ remain binary, whereas $R'_{\{\ttt{D},\ttt{K}\}}$ has disappeared (more precisely, if $R_{\{\ttt{D},\ttt{K}\}}$ does not contain a tuple taking values \ttt{d} and \ttt{k} on \ttt{D} and \ttt{K} respectively, then $\Q'$ has an empty answer; otherwise, we proceed in the way explained next).

\vgap 

Our algorithm solves the residual query $\Q'$ of Figure~\ref{fig&intro-join}c as follows: 
\begin{enumerate} 
    \item {\em Perform a semi-join reduction.} There are two steps. First, for every vertex $X$ in Figure~\ref{fig&intro-join}c, intersect all the unary relations over $\{X\}$ (if any) into a single list $R''_{\{X\}}$. For example, the two unary relations $R'_{\{\ttt{A}\}|\ttt{d}}$ and $R'_{\{\ttt{A}\}|\ttt{e}}$ of $\ttt{A}$ are intersected to produce $R''_{\{\ttt{A}\}}$; only the values in the intersection can appear in the join result. Second, for every non-isolated attribute $X$ in Figure~\ref{fig&intro-join}c, use $R''_{\{X\}}$ to shrink each binary relation $R'_{\{X, Y\}}$ (for all relevant $Y$) to eliminate tuples whose $X$-values are absent in $R''_{\{X\}}$. This reduces $R'_{\{X, Y\}}$ to a subset $R''_{\{X, Y\}}$. For example, every tuple in $R''_{\{\ttt{A},\ttt{B}\}}$ uses an $\ttt{A}$-value from $R''_{\{\ttt{A}\}}$ and a $\ttt{B}$-value from $R''_{\{\ttt{B}\}}$.
    
    \vgap
    
    \item {\em Compute a cartesian product.} The residual query $\Q'$ can now be further simplified into a join query $\Q''$ which includes (i) the relation $R''_{\{X\}}$ for every isolated attribute $X$, and (ii) the relation $R''_{\{X, Y\}}$ for every solid edge in Figure~\ref{fig&intro-join}c. Figure~\ref{fig&intro-join}d gives a neater view of $\Q''$; clearly, $\join(\Q'')$ is the cartesian product of $R''_{\{\ttt{G}\}}$, $R''_{\{\ttt{H}\}}$, $R''_{\{\ttt{L}\}}$, $R''_{\{\ttt{I}, \ttt{J}\}}$, and the result of the ``triangle join'' $\{R''_{\{\ttt{A}, \ttt{B}\}}, R''_{\{\ttt{A}, \ttt{C}\}}$, $R''_{\{\ttt{B}, \ttt{C}\}}\}$. 
\end{enumerate}

\vgap

As mentioned earlier, we plan to use only a small subset of the $p$ servers to compute $\Q'$. It turns out that the load of our strategy depends heavily on the cartesian product of the {\em unary} relations $R''_{\{X\}}$ (one for every isolated attribute $X$, i.e., $R''_{\{\ttt{G}\}}$, $R''_{\{\ttt{H}\}}$, and $R''_{\{\ttt{L}\}}$ in our example) in a configuration. Ideally, if the cartesian product of {\em every} configuration is small, we can prove a load of $\tO(m/p^{1/\rho})$ easily. Unfortunately, this is not true: in the worst case, the cartesian products of various configurations can differ dramatically. 

\begin{table}
  \resizebox{\textwidth}{!}{%
    \begin{tabular}{c|l|c}
        {\bf symbol} & {\bf meaning} & {\bf definition} \\ 
        \hline
        $p$ & number of machines & {Sec~\ref{sec&intro-prob}} \\ 
        $\Q$ & join query & {Sec~\ref{sec&intro-prob}} \\ 
        $m$ & input size of $\Q$ & {\eqref{eqn&inputsize}} \\
        $\join(\Q)$ & result of $\Q$ & {Sec~\ref{sec&intro-prob}} \\ 
        $\attset(\Q)$ & set of attributes in the relations of $\Q$ &  {Sec~\ref{sec&intro-prob}} \\
        $\G(\V, \E)$ & hypergraph of $\Q$ & {Sec~\ref{sec&hypergraph}} \\ 
        $W$ & fractional edge covering/packing of $\G$ & {Sec~\ref{sec&hypergraph}} \\ 
        $W(e)$ & weight of an edge $e \in \E$ & {Sec~\ref{sec&hypergraph}} \\ 
        $\rho$ (or $\tau$) & fractional edge covering (or packing) number of $\G$ & {Sec~\ref{sec&hypergraph}} \\  
        $R_e$ ($e \in \E$) & relation $R \in \Q$ with $\scheme(R) = e$ & {Sec~\ref{sec&hypergraph}} \\
        $\lambda$ & heavy parameter & {Sec~\ref{sec&taxonomy}} \\ 
        $\H$ & set of heavy attributes in $\attset(\Q)$ & {Sec~\ref{sec&taxonomy}} \\ 
        $\config(\Q,\H)$ & set of configurations of $\H$ & {Sec~\ref{sec&taxonomy}} \\ 
        $\bm{\eta}$ & configuration & {Sec~\ref{sec&taxonomy}} \\ 
        $R'_e(\bm{\eta})$ & residual relation of $e \in \E$ under $\bm{\eta}$ & {Sec~\ref{sec&taxonomy}} \\ 
        $\Q'(\bm{\eta})$ & residual query under $\bm{\eta}$ & {\eqref{eqn&Q'-eta}} \\ 
        $k$ & size of $\attset(\Q)$ & {Lemma~\ref{lmm&taxonomy-inputsize}} \\ 
        $m_{\bm{\eta}}$ & input size of $\Q'(\bm{\eta})$ & {Lemma~\ref{lmm&taxonomy-inputsize}} \\ 
        $\L$ & set of light attributes in $\attset(\Q)$ & \eqref{eqn&L} \\
        $\I$ & set of isolated attributes in $\attset(\Q)$ & \eqref{eqn&I} \\
        $R''_X(\bm{\eta})$ & relation on attribute $X$ after semi-join reduction & \eqref{eqn&R''-X} \\
        $R''_e(\bm{\eta})$ & relation on $e \in \E$ after semi-join reduction & {Sec~\ref{sec&framework-semi}}\\
        $\Q''_\iso(\bm{\eta})$ & query on the isolated attributes after semi-join reduction & \eqref{eqn&Q''-iso} \\
        $\Q''_\mit{light}(\bm{\eta})$ & query on the light edges after semi-join reduction & \eqref{eqn&Q''-light} \\
        $\Q''(\bm{\eta})$ & reduced query under $\bm{\eta}$ & \eqref{eqn&Q''} \\
        $W_\I$ & total weight of all vertices in $\I$ under fractional edge packing $W$ & \eqref{eqn&W_I} \\ 
        $\J$ & non-empty subset of $\I$ & {Sec~\ref{sec&cpthm2}} \\
        $\Q''_\J(\bm{\eta})$ & query on the isolated attributes in $\J$ after semi-join reduction & \eqref{eqn&Q''_J} \\
        $W_\J$ & total weight of all vertices in $\J$ under fractional edge packing $W$ & \eqref{eqn&W_J}   
    \end{tabular}%
  }
    \caption{Frequently used notations} 
    \label{tab&symbols}
\end{table}

\vgap

Our isolated cartesian product theorem (Theorem~\ref{thm&cpthm}) shows that the cartesian product size is small when {\em averaged over all the possible configurations}. This property allows us to allocate a different number of machines to process each configuration in parallel while ensuring that the total number of machines required will not exceed $p$. The theorem is of independent interest and may be useful for developing join algorithms under other computation models (e.g., the external memory model \cite{av88}; see Section~\ref{sec&conclusion}). 

\subsection{An Application: Subgraph Enumeration}

The joins studied in this paper bear close relevance to the {\em subgraph enumeration problem}, where the goal is to find all occurrences of a pattern subgraph $G' = (V', E')$ in a graph $G = (V, E)$. This problem is NP-hard \cite{c71} if the size of $G'$ is unconstrained, but is polynomial-time solvable when $G'$ has only a constant number of vertices. In the MPC model, the edges of $G$ are evenly distributed onto the $p$ machines at the beginning, whereas an algorithm must produce every occurrence on at least one machine in the end. The following facts are folklore regarding a constant-size $G'$: 
\begin{itemize}
    \item Every constant-round subgraph enumeration algorithm must incur a load of $\Omega(|E|/p^{1/\rho})$,\footnote{Here, we consider $|E| \ge |V|$ because vertices with no edges can be discarded directly.} where $\rho$ is the fractional edge covering number (Section~\ref{sec&hypergraph}) of $G'$. 
    
    \item The subgraph enumeration problem can be converted to a simple binary join with input size $O(|E|)$ and the same fractional edge covering number $\rho$. 
\end{itemize}
Given a constant-size $G'$, our join algorithm (Theorem~\ref{thm&alg-main}) solves subgraph enumeration with load $\tO(|E|/p^{1/\rho})$, which is optimal up to a polylogarithmic factor. 

\subsection{Remarks.} This paper is an extension of \cite{ks17} and \cite{t20}. Ketsman and Suciu \cite{ks17} were the first to discover a constant-round algorithm to solve simple binary joins with an asymptotically optimal load. Tao \cite{t20} introduced a preliminary version of the isolated cartesian product theorem and applied it to simplify the algorithm of \cite{ks17}. The current work features a more powerful version of the isolated cartesian product theorem (see the remark in Section~\ref{sec&cpthm-corollary}). Table~\ref{tab&symbols} lists the symbols that will be frequently used. 

\section{Hypergraphs and the AGM Bound} \label{sec&hypergraph}

We define a {\em hypergraph} $\G$ as a pair $(\V, \E)$ where: 
\begin{itemize} 
    \item $\V$ is a finite set, where each element is called a {\em vertex}; 
    \item $\E$ is a set of subsets of $\V$, where each subset is called a {\em (hyper-)edge}.
\end{itemize}
An edge $e$ is {\em unary} or {\em binary} if $|e| = 1$ or $2$, respectively. $\G$ is {\em binary} if all its edges are binary. 

\vgap

Given a vertex $X \in \V$ and an edge $e \in \E$, we say that $X$ and $e$ are {\em incident} to each other if $X \in e$. Two distinct vertices $X, Y \in \V$ are {\em adjacent} if there is an $e \in \E$ containing $X$ and $Y$. All hypergraphs discussed in this paper have the property that every vertex is incident to at least one edge. 

\vgap

Given a subset $\U$ of $\V$, we define the subgraph {\em induced} by $\U$ as $(\U, \E_\U)$ where $ \E_\U = \{ \U \cap e \bigm| e \in \E\}.$
    
\extraspacing {\bf Fractional Edge Coverings and Packings.} 
Let $\G = (\V, \E)$ be a hypergraph and $W$ be a function mapping $\E$ to real values in $[0, 1]$. We call $W(e)$ the {\em weight} of edge $e$ and $\sum_{e \in \E} W(e)$ the {\em total weight} of $W$. Given a vertex $X \in \V$, we refer to $\sum_{e \in \E : X \in e} W(e)$ (i.e., the sum of the weights of all the edges incident to $X$) as the {\em weight} of $X$.  

\vgap

$W$ is a {\em fractional edge covering} of $\G$ if the weight of every vertex $X \in \V$ is at least 1. The {\em fractional edge covering number} of $\G$ --- denoted as $\rho(\G)$ --- equals the smallest total weight of all the fractional edge coverings. $W$ is a {\em fractional edge packing} if the weight of every vertex $X \in \V$ is at most 1. The {\em fractional edge packing number} of $\G$ --- denoted as $\tau(\G)$ --- equals the largest total weight of all the fractional edge packings. A fractional edge packing $W$ is {\em tight} if it is simultaneously also a fractional edge covering; likewise, a fractional edge covering $W$ is {\em tight} if it is simultaneously also a fractional edge packing. Note that in a tight fractional edge covering/packing, the weight of every vertex must be exactly 1.

\vgap

Binary hypergraphs have several interesting properties:

\begin{lem} \label{lmm&coverpack}
    If $\G$ is binary, then:
    \begin{itemize}
        \item $\rho(\G) + \tau(\G) = |\V|$; furthermore, $\rho(\G) \ge \tau(\G)$, where the equality holds if and only if $\G$ admits a tight fractional edge packing (a.k.a.\ tight fractional edge covering).
        \item $\G$ admits a fractional edge packing $W$ of total weight $\tau(\G)$ such that 
        \begin{enumerate}
            \item the weight of every vertex $X \in \V$ is either 0 or 1;  
            \item if $\Z$ is the set of vertices in $\V$ with weight 0, then $\rho(\G) - \tau(\G) = |\Z|$. 
        \end{enumerate}
    \end{itemize}
\end{lem}

\begin{proof}
    The first bullet is proved in Theorem 2.2.7 of \cite{su97}. The fractional edge packing $W$ in Theorem 2.1.5 of \cite{su97} satisfies Property (1) of the second bullet. Regarding such a $W$, we have 
    $$\tau(\G) = \textrm{total weight of $W$} = \frac{1}{2} \sum_{X \in \V} (\textrm{weight of $X$}) = (|\V| - |\Z|)/2.$$ Plugging this into $\rho(\G) + \tau(\G) = |\V|$ yields $\rho(\G) = (|\V| + |\Z|)/2$. Hence, Property (2) follows. 
\end{proof}

\extraspacing {\bf {\em Example.}} Suppose that $\G$ is the binary hypergraph in Figure~\ref{fig&intro-join}a. It has a fractional edge covering number $\rho(\G) = 6.5$, as is achieved by the function $W_1$ that maps \{\ttt{G}, \ttt{F}\}, \{\ttt{D}, \ttt{K}\}, \{\ttt{I}, \ttt{J}\}, \{\ttt{E}, \ttt{H}\}, and \{\ttt{E}, \ttt{L}\} to 1, \{\ttt{A}, \ttt{B}\}, \{\ttt{A}, \ttt{C}\}, and \{\ttt{B}, \ttt{C}\} to $1/2$, and the other edges to 0. Its fractional edge packing number is $\tau(\G) = 5.5$, achieved by the function $W_2$ which is the same as $W_1$ except that $W_2$ maps \{\ttt{E}, \ttt{L}\} to 0. Note that $W_2$ satisfies both properties of the second bullet (here $\Z = \{\ttt{L}\}$). 
\done

\extraspacing {\bf Hypergraph of a Join Query and the AGM Bound.}
Every join $\Q$ defines a hypergraph $\G = (\V, \E)$ where $\V = \attset(\Q)$ and $\E = \{ \scheme(R) \bigm| R \in \Q \}$. When $\Q$ is simple, for each edge $e \in \E$ we denote by $R_e$ the input relation $R \in \Q$ with $e = \scheme(R)$. The following result is known as the {\em AGM bound}:

\begin{lemC}[\cite{agm13}] \label{lmm&agm}
    Let $\Q$ be a simple binary join and $W$ be any fractional edge covering of the hypergraph $\G = (\V, \E)$ defined by $\Q$. Then, $     |\join(\Q)|
        \le
        \prod_{e \in \E} |R_e|^{W(e)}.$
\end{lemC}

The fractional edge covering number of $\Q$ equals $\rho(\G)$ and, similarly, the fractional edge packing number of $\Q$ equals $\tau(\G)$. 

\extraspacing {\bf Remark on the Fractional Edge Quasi-Packing Number.} Although the technical development in the subsequent sections is irrelevant to ``fractional edge quasi-packing number'', we provide a full definition of the concept here because it enables the reader to better distinguish our solution and the one-round algorithm of \cite{kbs16} (reviewed in Section~\ref{sec&intro-prev}). Consider a hypergraph $\G = (\V, \E)$. For each subset $\U \subseteq \V$, let $\G_{\setminus \U}$ be the graph obtained by removing $\U$ from all the edges of $\E$, or formally: $\G_{\setminus \U} = (\V \setminus \U, \E_{\setminus \U})$ where $\E_{\setminus \U} = \{e \setminus \U \mid e \in \E \textrm{ and } e \setminus \U \ne \emptyset\}$. The {\em fractional edge quasi-packing number} of $\G$ --- denoted as $\psi(\G)$ --- is 
\myeqn{
    \psi(\G) &=& \max_{\textrm{all $\U \subseteq \V$}} \tau(\G_{\setminus \U}) \nn 
}
where $\tau(\G_{\setminus \U})$ is the fractional edge packing number of $\G_{\setminus \U}$.

\vgap 

In \cite{kbs16}, Koutris, Beame, and Suciu proved that $\psi(\G) \ge \rho(\G)$ holds on any $\G$ (which need not be binary). In general, $\psi(\G)$ can be considerably higher than $\rho(\G)$. In fact, this is true even on ``regular'' binary graphs, about which we mention two examples (both can be found in \cite{kbs16}): 
\begin{itemize}
    \item when $\G$ is a clique, $\psi(\G) = |\V| - 1$ but $\rho(\G)$ is only $|\V|/2$; 
    \item when $\G$ is a cycle, $\psi(\G) = \lc 2(|\V| - 1)/3 \rc$ and $\rho(\G)$ is again $|\V|/2$.
\end{itemize}

\vgap

If $\G$ is the hypergraph defined by a query $\Q$, $\psi(\G)$ is said to be the query's fractional edge covering number. It is evident from the above discussion that, when $\G$ is a clique or a cycle, the load $\tO(m/p^{1/\rho(\G)})$ of our algorithm improves the load $\tO(m/p^{1/\psi(\G)})$ of \cite{kbs16} by a polynomial factor.

\section{Fundamental MPC Algorithms} \label{sec&mpcbasic} 

This subsection will discuss several building-block routines in the MPC model useful later. 

\extraspacing {\bf Cartesian Products.} 
Suppose that $R$ and $S$ are relations with disjoint schemes. Their {\em cartesian product}, denoted as $R \times S$, is a relation over $\scheme(R) \cup \scheme(S)$ that consists of all the tuples $\bm{u}$ over $\scheme(R) \cup \scheme(S)$ such that $\bm{u}[\scheme(R)] \in R$ and $\bm{u}[\scheme(S)] \in S$. 

\vgap 

The lemma below gives a deterministic algorithm for computing the cartesian product: 

\begin{lem} \label{lmm&cpalg}
    Let $\Q$ be a set of $t = O(1)$ relations $R_1, R_2, ..., R_t$ with disjoint schemes. The tuples in $R_i$ ($1 \le i \le t$) have been labeled with ids 1, 2, ..., $|R_i|$, respectively. We can deterministically compute $\cp(\Q) = R_1 \times R_2 \times ... \times R_t$ in one round with load 
    \myeqn{
    O\left(
    \max_{\textrm{non-empty $\Q' \subseteq \Q$}} 
    \fr{|\cp(\Q')|^\fr{1}{|\Q'|}}{p^\fr{1}{|\Q'|}}
    \right) \label{eqn&cp-load1}
    } 
    using $p$ machines. Alternatively, if we assume $|R_1| \ge |R_2| \ge ... \ge |R_t|$, then the load can be written as 
    \myeqn{
    O\left(
    \max_{i=1}^t  
    \fr{|\cp(\myset{R_1, R_2, ..., R_i})|^\fr{1}{i}}{p^\fr{1}{i}}
    \right). \label{eqn&cp-load2}
    }
    In \eqref{eqn&cp-load1} and \eqref{eqn&cp-load2}, the constant factors in the big-$O$ depend on $t$. 
\end{lem}

\begin{proof}
    For each $i \in [1, t]$, define $\Q_i = \myset{R_1, ..., R_i}$ and $L_i = |\cp(\Q_i)|^\fr{1}{i} / p^\fr{1}{i}$. 
    Let $t'$ be the largest integer satisfying $|R_i| \ge L_i$ for all $i \in [1, t']$; $t'$ definitely exists because $|R_1| \ge L_1 = |R_1|/p$. Note that this means $|R_t| \le |R_{t-1}| \le ... \le |R_{t'+1}| < L_{t'+1}$ if $t' < t$. 
    
    \vgap
    
    Next, we will explain how to obtain $\cp(\Q_{t'})$ with load $O(L_{t'})$. If $t' < t$, this implies that $\cp(\Q)$ can be obtained with load $O(L_{t'} + L_{t'+1})$ because $R_{t'+1}, ..., R_t$ can be broadcast to all the machines with an extra load $O(L_{t'+1} \cdot (t-t')) = O(L_{t'+1})$. 
    
    \vgap 
    
    Align the machines into a $t'$-dimensional $p_1 \times p_2 \times ... \times p_{t'}$ grid where $$p_i = \lf |R_i| / L_{t'} \rf$$ for each $i \in [1, t']$. This is possible because $|R_i| \ge |R_{t'}| \ge L_{t'}$ and $ \prod_{i=1}^{t'} \fr{|R_i|}{L_{t'}} = \fr{|\cp(\Q_{t'})|}{(L_{t'})^{t'}} = p$. Each machine can be uniquely identified as a $t'$-dimensional point $(x_1, ..., x_{t'})$ in the grid where $x_i \in [1, p_i]$ for each $i \in [1, t']$. For each $R_i$, we send its tuple with id $j \in [1, |R_i|]$ to all the machines whose coordinates on dimension $i$ are $(j\mod p_i) + 1$. Hence, a machine receives $O(|R_i| / p_i) = O(L_{t'})$ tuples from $R_i$; and the overall load is $O(L_{t'} \cdot t') = O(L_{t'})$. For each combination of $\bm{u}_1,  \bm{u}_2, ..., \bm{u}_{t'}$ where $\bm{u}_i \in R_i$, some machine has received all of $\bm{u}_1, ..., \bm{u}_{t'}$. Therefore, the algorithm is able to produce the entire $\cp(\Q_{t'})$. 
\end{proof}

The load in \eqref{eqn&cp-load2} matches a lower bound stated in Section 4.1.5 of \cite{kss18}. The algorithm in the above proof generalizes an algorithm in \cite{hyt19} for computing the cartesian product of $t=2$ relations. The randomized hypercube algorithm of \cite{bks17c} incurs a load higher than \eqref{eqn&cp-load2} by a logarithmic factor and can fail with a small probability. 

\extraspacing {\bf Composition by Cartesian Product.} If we already know how to solve queries $\Q_1$ and $\Q_2$ separately, we can compute the cartesian product of their results efficiently:

\begin{lem} \label{lmm&cpcomp}
    Let $\Q_1$ and $\Q_2$ be two join queries satisfying the condition $\attset(\Q_1) \cap \attset(\Q_2) = \emptyset$. Let $m$ be the total number of tuples in the input relations of $\Q_1$ and $\Q_2$. Suppose that
    \begin{itemize}
        \item with probability at least $1 - \delta_1$, we can compute in one round $\join(\Q_1)$ with load $\tO(m / p_1^{1/t_1})$ using $p_1$ machines;
        
        \item with probability at least $1 - \delta_2$, we can compute in one round $\join(\Q_2)$ with load $\tO(m / p_2^{1/t_2})$ using $p_2$ machines.
    \end{itemize}
    Then, with probability at least $1 - \delta_1 - \delta_2$, we can compute $\join(\Q_1) \times \join(\Q_2)$ in one round with load
    $\tO(\max\{m / p_1^{1/t_1}, m / p_2^{1/t_2}\})$ 
    using $p_1 p_2$ machines.
\end{lem}

\begin{proof}
    Let $\A_1$ and $\A_2$ be the algorithm for $\Q_1$ and $\Q_2$, respectively. If a tuple $\bm{u} \in \join(\Q_1)$ is produced by $\A_1$ on the $i$-th ($i \in [1, p_1]$) machine, we call $\bm{u}$ an {\em $i$-tuple}. Similarly, if a tuple $\bm{v} \in \join(\Q_2)$ is produced by $\A_2$ on the $j$-th ($j \in [1, p_2]$) machine, we call $\bm{v}$ a {\em $j$-tuple}.
    
    \vgap
    
    Arrange the $p_1p_2$ machines into a matrix where each row has $p_1$ machines and each column has $p_2$ machines (note that the number of rows is $p_2$ while the number of columns is $p_1$). For each row, we run $A_1$ using the $p_1$ machines on that row to compute $\join(\Q_1)$; this creates $p_2$ instances of $A_1$ (one per row). If $A_1$ is randomized, we instruct all those instances to take the same random choices.\footnote{The random choices of an algorithm can be modeled as a sequence of random bits. Once the sequence is fixed, a randomized algorithm becomes deterministic. An easy way to ``instruct'' all instances of $A_1$ to make the same random choices is to ask all the participating machines to pre-agree on the random-bit sequence. For example, one machine can generate all the random bits and send them to the other machines. Such communication happens {\em before} receiving $\Q$ and hence does not contribute to the query's load. The above approach works for a single $\Q$ (which suffices for proving Lemma~\ref{lmm&cpcomp}). There is a standard technique \cite{n91} to extend the approach to work for any number of queries. The main idea is to have the machines pre-agree on a sufficiently large number of random-bit sequences. Given a query, a machine randomly picks a specific random-bit sequence and broadcasts the sequence's id (note: only the id, not the sequence itself) to all machines. As shown in \cite{n91}, such an id can be encoded in $\tO(1)$ words. Broadcasting can be done in constant rounds with load $O(p^\eps)$ for an arbitrarily small constant $\eps > 0$.
    } This ensures:
    \begin{itemize}
        \item with probability at least $1-\delta_1$, all the instances succeed simultaneously; 
        \item for each $i \in [1, p_1]$, all the machines at the $i$-th column produce exactly the same set of $i$-tuples. 
    \end{itemize}
    The load incurred is  $\tO(m / p_1^{1/t_1})$.    
    Likewise, for each column, we run $A_2$ using the $p_2$ machines on that column to compute $\join(\Q_2)$. With probability at least $1 - \delta_2$, for each $j \in [1, p_2]$, all the machines at the $j$-th row produce exactly the same set of $j$-tuples. The load is  $\tO(m / p_2^{1/t_2})$.
    
    \vgap
    
    Therefore, it holds with probability at least $1 - \delta_1 - \delta_2$ that, for each pair $(i, j)$, some machine has produced all the $i$- and $j$-tuples. Hence, every tuple of $\join(\Q_1) \times \join(\Q_2)$ appears on a machine. The overall load is the larger between $\tO(m / p_1^{1/t_1})$ and $\tO(m / p_2^{1/t_2})$.
\end{proof}

 \extraspacing {\bf Skew-Free Queries.} It is possible to solve a join query $\Q$ on binary relations in a single round with a small load if no value appears too often. To explain, denote by $m$ the input size of $\Q$; set $k = |\attset(\Q)|$, and list out the attributes in $\attset(\Q)$ as $X_1, ..., X_k$. For $i \in [1, k]$, let $p_i$ be a positive integer referred to as the {\em share} of $X_i$. A relation $R \in \Q$ with scheme $\{X_i, X_j\}$ is {\em skew-free} if every value $x \in \dom$ fulfills both conditions below: 
\begin{itemize}
    \item $R$ has $O(m/p_i)$ tuples $\bm{u}$ with $\bm{u}(X_i) = x$;
    \item $R$ has $O(m/p_j)$ tuples $\bm{u}$ with $\bm{u}(X_j) = x$.
\end{itemize}
Define $\share(R) = p_i \cdot p_j$. If every $R \in \Q$ is skew-free, $\Q$ is {\em skew-free}. We know: 

\begin{lemC}[\cite{bks17c}] \label{lmm&skewfree}
    With probability at least $1 - 1/p^c$ where $p = \prod_{i=1}^k p_i$ and $c \ge 1$ can be set to an arbitrarily large constant, a skew-free query $\Q$ with input size $m$ can be answered in one round with load $\tO(m / \min_{R\in \Q} \share(R))$ using $p$ machines. 
\end{lemC}

\section{A Taxonomy of the Join Result} \label{sec&taxonomy}

Given a simple binary join $\Q$, we will present a method to partition $\join(\Q)$ based on the value frequencies in the relations of $\Q$. Denote by $\G = (\V, \E)$ the hypergraph defined by $\Q$ and by $m$ the input size of $\Q$. 

\extraspacing {\bf Heavy and Light Values.} Fix an arbitrary integer $\lambda \in [1, m]$. A value $x \in \dom$ is 
\begin{itemize} 
    \item {\em heavy} if $| \{ \bm{u} \in R \bigm| \bm{u}(X) = x \} | \ge m/\lambda$ for some relation $R \in \Q$ and some attribute $X \in \scheme(R)$;
  
    \item {\em light} if $x$ is not heavy, but appears in at least one relation $R \in \Q$.
\end{itemize}
It is easy to see that each attribute has at most $\lambda$ heavy values. Hence, the total number of heavy values is at most $\lambda \cdot |\attset(\Q)| = O(\lambda)$. We will refer to $\lambda$ as the {\em heavy parameter}. 

 \extraspacing {\bf Configurations.} Let $\H$ be an arbitrary (possibly empty) subset of $\attset(\Q)$. A {\em configuration} of $\H$ is a tuple $\bm{\eta}$ over $\H$ such that $\bm{\eta}(X)$ is heavy for every $X \in \H$. Let $\config(\Q, \H)$ be the set of all configurations of $\H$. It is clear that $|\config(\Q, \H)| = O(\lambda^{|\H|})$.

\extraspacing {\bf Residual Relations/Queries.} Consider an edge $e \in \E$; define $e' = e \setminus \H$. We say that $e$ is {\em active} on $\H$ if $e' \neq \emptyset$, i.e., $e$ has at least one attribute outside $\H$. 
An active $e$ defines a {\em residual relation} under $\bm{\eta}$ --- denoted as $R'_{e}(\bm{\eta})$ --- which
\begin{itemize}
    \item is over $e'$ and 
    \item consists of every tuple $\bm{v}$ that is the projection (on $e'$) of some tuple $\bm{w} \in R_e$ ``consistent'' with $\bm{\eta}$, namely: 
      \begin{itemize}
        \item $\bm{w}(X) = \bm{\eta}(X)$ for every $X \in e \cap \H$;
        \item $\bm{w}(Y)$ is light for every $Y \in e'$;
        \item $\bm{v} = \bm{w}[e']$.  
    \end{itemize}
\end{itemize}

The {\em residual query} under $\bm{\eta}$ is
\myeqn{
    \Q'(\bm{\eta}) = \left\{R'_e(\bm{\eta}) \bigm| e \in \E, \textrm{ $e$ active on $\H$ } \right\}.
    \label{eqn&Q'-eta}
}
Note that if $\H = \attset(\Q)$, $\Q'(\bm{\eta})$ is empty.

  \extraspacing {\bf {\em Example.}} Consider the query $\Q$ in Section~\ref{sec&intro-ours} (hypergraph $\G$ in Figure~\ref{fig&intro-join}a) and the configuration $\bm{\eta}$ of $\H = \{\ttt{D}, \ttt{E}, \ttt{F}, \ttt{K}\}$ where $\bm{\eta}[\ttt{D}] = \ttt{d}$, $\bm{\eta}[\ttt{E}] = \ttt{e}$, $\bm{\eta}[\ttt{F}] = \ttt{f}$, and $\bm{\eta}[\ttt{K}] = \ttt{k}$. If $e$ is the edge $\{\ttt{A}, \ttt{D}\}$, then $e' = \{\ttt{A}\}$ and $R'_{e}(\bm{\eta})$ is the relation $R'_{\{\ttt{A}\} \mid \ttt{d}}$ mentioned in Section~\ref{sec&intro-ours}. If $e$ is the edge $\{\ttt{A}, \ttt{B}\}$, then $e' = \{\ttt{A}, \ttt{B}\}$ and $R'_{e}(\bm{\eta})$ is the relation $R'_{\{\ttt{A}, \ttt{B}\}}$ in Section~\ref{sec&intro-ours}. The residual query $\Q'(\bm{\eta})$ is precisely the query $\Q'$ in Section~\ref{sec&intro-ours}.
\done

\vgap

It is rudimentary to verify
\myeqn{
        \join(\Q)
        =
        \bigcup_{\H} \Big(\bigcup_{\bm{\eta} \in \config(\Q, \H)} \join(\Q'(\bm{\eta})) \times \{\bm{\eta}\}\Big). 
        \label{eqn&taxonomy-goal}
}

\begin{lem} \label{lmm&taxonomy-inputsize}
    Let $\Q$ be a simple binary join with input size $m$ and $\H$ be a subset of $\attset(\Q)$. For each configuration $\bm{\eta} \in \config(\Q, \H)$, denote by $m_\bm{\eta}$ the total size of all the relations in $\Q'(\bm{\eta})$. We have:
    \myeqn{
        \sum_{\bm{\eta} \in \config(\Q, \H)} m_\bm{\eta}
        &\le& 
        m \cdot \lambda^{k-2} 
        \nn
    }
    where $k = |\attset(\Q)|$. 
\end{lem}

\begin{proof}
    Let $e$ be an edge in $\E$ and fix an arbitrary tuple $\bm{u} \in R_e$. Tuple $\bm{u}$ contributes 1 to the term $m_\bm{\eta}$ only if $\bm{\eta}(X) = \bm{u}(X)$ for every attribute $X \in e \cap \H$. How many such configurations $\bm{\eta}$ can there be? As these configurations must have the same value on every attribute in $e \cap \H$, they can differ only in the attributes of $\H \setminus e$. Since each attribute has at most $\lambda$ heavy values, we conclude that the number of those configurations $\eta$ is at most $\lambda^{|\H \setminus e|}$. $|\H \setminus e|$ is at most $k - 2$ because $|\H| \le k$ and $e$ has two attributes. The lemma thus follows. 
\end{proof}

\section{A Join Computation Framework} \label{sec&framework}

Answering a simple binary join $\Q$ amounts to producing the right-hand side of \eqref{eqn&taxonomy-goal}. Due to symmetry, it suffices to explain how to do so for an arbitrary subset $\H \subseteq \attset(\Q)$, i.e., the computation of 
\myeqn{
    \bigcup_{\bm{\eta} \in \config(\Q, \H)} \join(\Q'(\bm{\eta})).  
    \label{eqn&framework-goal}
}
At a high level, our strategy (illustrated in Section~\ref{sec&intro-ours}) works as follows. Let $\G = (\V, \E)$ be the hypergraph defined by $\Q$. We will remove the vertices in $\H$ from $\G$, which disconnects $\G$ into connected components (CCs). We divide the CCs into two groups: (i) the set of CCs each involving at least 2 vertices, and (ii) the set of all other CCs, namely those containing only 1 vertex. We will process the CCs in Group 1 {\em together} using Lemma~\ref{lmm&skewfree}, process the CCs in Group 2 {\em together} using Lemma~\ref{lmm&cpalg}, and then compute the cartesian product between Groups 1 and 2 using Lemma~\ref{lmm&cpcomp}.

\vgap

Sections~\ref{sec&framework-removeH} and \ref{sec&framework-semi} will formalize the strategy into a processing framework. Sections~\ref{sec&cpthm} and \ref{sec&cpthm2} will then establish two important properties of this framework, which are the key to its efficient implementation in Section~\ref{sec&alg}. 

\subsection{Removing the Attributes in \texorpdfstring{$\bm{\H}$}{H}} \label{sec&framework-removeH}

We will refer to each attribute in $\H$ as a {\em heavy attribute}. Define 
\myeqn{
    \L &=& \attset(\Q) \setminus \H; 
    \label{eqn&L}
}
and call each attribute in $\L$ a {\em light attribute}. An edge $e \in \E$ is 
\begin{itemize}
    \item a {\em light edge} if $e$ contains two light attributes, or
    \item a {\em cross edge} if $e$ contains a heavy attribute and a light attribute. 
\end{itemize}
A light attribute $X \in \L$ is a {\em border attribute} if it appears in at least one cross edge $e$ of $\G$. Denote by $\G' = (\L, \E')$ the subgraph of $\G$ induced by $\L$. A vertex $X \in \L$ is {\em isolated} if $\{X\}$ is the only edge in $\E'$ incident to $X$. Define 
\myeqn{
    \I 
    &=& 
    \myset{X \in \L \mid \textrm{$X$ is isolated}}. 
    \label{eqn&I}
}

\extraspacing {\bf {\em Example (cont.).}} Consider again the join query $\Q$ whose hypergraph $\G$ is shown in Figure~\ref{fig&intro-join}a. Set $\H = \{\ttt{D}, \ttt{E}, \ttt{F}, \ttt{K}\}$. Set $\L$ includes all the white vertices in Figure~\ref{fig&intro-join}b. Edge $\{\ttt{A}, \ttt{B}\}$ is a light edge, $\{\ttt{A}, \ttt{D}\}$ is a cross edge, while $\{\ttt{D}, \ttt{K}\}$ is neither a light edge nor a cross edge. All the vertices in $\L$ except $\ttt{J}$ are border vertices. Figure~\ref{fig&L} shows the subgraph of $\G$ induced by $\L$, where a unary edge is represented by a box and a binary edge by a segment. The isolated vertices are $\ttt{G}$, $\ttt{H}$, and \ttt{L}.    
\done

\subsection{Semi-Join Reduction} \label{sec&framework-semi}

Recall from Section~\ref{sec&taxonomy} that every configuration $\bm{\eta}$ of $\H$ defines a residual query $\Q'(\bm{\eta})$. Next, we will simplify $\Q'(\bm{\eta})$ into a join $\Q''(\bm{\eta})$ with the same result.  

\vgap

Observe that the hypergraph defined by $\Q'(\bm{\eta})$ is always $\G' = (\L, \E')$, regardless of $\bm{\eta}$. Consider a border attribute $X \in \L$ and a cross edge $e$ of $\G = (\V,\E)$ incident to $X$. As explained in Section~\ref{sec&taxonomy}, the input relation $R_e \in \Q$ defines a unary residual relation $R'_e(\bm{\eta}) \in \Q'(\bm{\eta})$. Note that $R'_e(\bm{\eta})$ has scheme $\{X\}$. We define: 
\myeqn{
    R''_X (\bm{\eta})
    &=&
    \bigcap_{\textrm{cross edge $e \in \E$ s.t.\ $X \in e$}} R'_{e}(\bm{\eta}). \label{eqn&R''-X}
}

\begin{figure}
    \includegraphics[height=35mm]{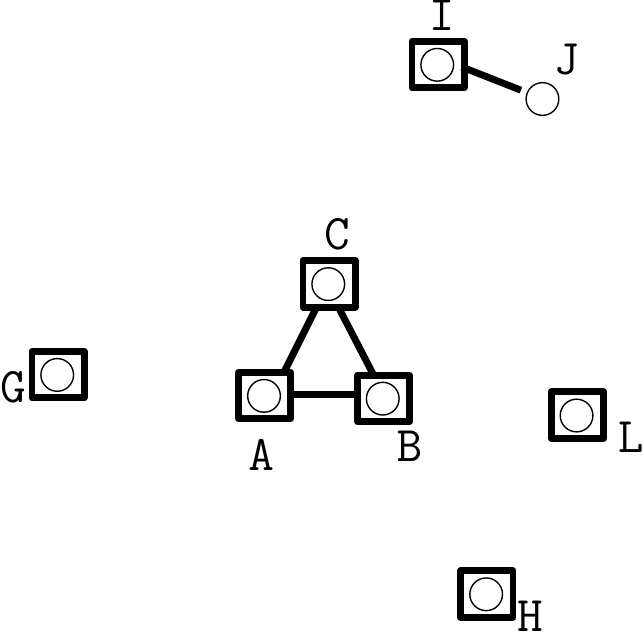} 
    \figcapup 
    \caption{Subgraph induced by $\L$} 
    \label{fig&L}
\end{figure}    

 \extraspacing {\bf {\em Example (cont.).}} Let $\H = \{\ttt{D}, \ttt{E}, \ttt{F}, \ttt{K}\}$, and consider its configuration $\bm{\eta}$ with $\bm{\eta}(\ttt{D}) = \ttt{d}$, $\bm{\eta}(\ttt{E}) = \ttt{e}$, $\bm{\eta}(\ttt{F}) = \ttt{f}$, and $\bm{\eta}(\ttt{K}) = \ttt{k}$. Set $X$ to the border attribute $\ttt{A}$. When $e$ is $\{\ttt{A}, \ttt{D}\}$ or $\{\ttt{A}, \ttt{E}\}$, $R'_e(\bm{\eta})$ is the relation $R'_{\{\ttt{A}\} \mid \ttt{d}}$ or $R'_{\{\ttt{A}\} \mid \ttt{e}}$ mentioned in Section~\ref{sec&intro-ours}, respectively. $\{\ttt{A}, \ttt{D}\}$ and $\{\ttt{A}, \ttt{E}\}$ are the only cross edges containing $\ttt{A}$. Hence, $R''_A(\bm{\eta}) = R'_{\{\ttt{A}\} \mid \ttt{d}} \cap R'_{\{\ttt{A}\} \mid \ttt{e}}$, which is the relation $R''_{\{A\}}$ in Section~\ref{sec&intro-ours}.
\done

\vgap 

Recall that every light edge $e = \{X, Y\}$ in $\G$ defines a residual relation $R'_{e}(\bm{\eta})$ with scheme $e$. We define $R''_{e}(\bm{\eta})$ as a relation over $e$ that contains every tuple $\bm{u} \in R'_{e}(\bm{\eta})$ satisfying:
\begin{itemize} 
    \item (applicable only if $X$ is a border attribute) $\bm{u}(X) \in R''_X (\bm{\eta})$; 
    \item (applicable only if $Y$ is a border attribute) $\bm{u}(Y) \in R''_Y (\bm{\eta})$.
\end{itemize}
Note that if neither $X$ nor $Y$ is a border attribute, then $R''_e(\bm{\eta}) = R'_e(\bm{\eta})$.

 \extraspacing {\bf {\em Example (cont.).}} For the light edge $e = \{\ttt{A}, \ttt{B}\}$, $R'_e(\bm{\eta})$ is the relation $R'_{\{\ttt{A}, \ttt{B}\}}$ mentioned in Section~\ref{sec&intro-ours}. Because $\ttt{A}$ and $\ttt{B}$ are border attributes, $R''_e(\bm{\eta})$ includes all the tuples in $R'_{\{\ttt{A}, \ttt{B}\}}$ that take an $\ttt{A}$-value from $R''_\ttt{A}(\bm{\eta})$ and a $\ttt{B}$-value from $R''_\ttt{B}(\bm{\eta})$. This $R''_e(\bm{\eta})$ is precisely the relation $R''_{\{\ttt{A}, \ttt{B}\}}$ in Section~\ref{sec&intro-ours}.
\done

\vgap 

Every vertex $X \in \I$ must be a border attribute and, thus, must now be associated with $R''_X(\bm{\eta})$. We can legally define: 
\myeqn{
    \Q''_\iso(\bm{\eta})
    &=&
    \{ R''_X (\bm{\eta}) \mid X \in \I \} \label{eqn&Q''-iso} \\ 
    \Q''_\mit{light} (\bm{\eta})
    &=&
    \{ R''_e(\bm{\eta}) \mid \textrm{light edge $e \in \E$} \} \label{eqn&Q''-light} \\ 
    \Q''(\bm{\eta})
    &=&
    \Q''_\mit{light}(\bm{\eta}) \cup \Q''_\iso(\bm{\eta}). \label{eqn&Q''}
}

\extraspacing {\bf {\em Example (cont.).}} $\Q''_\iso(\bm{\eta})=\{R''_{\{\ttt{G} \}}$, $R''_{\{\ttt{H}\}}$, $R''_{\{\ttt{L}\}}\}$ and $\Q''_{\mit{light}}(\bm{\eta}) = \{R''_{\{\ttt{A}, \ttt{B}\}}$, $R''_{\{\ttt{A}, \ttt{C}\}}$, $R''_{\{\ttt{B}, \ttt{C}\}}$, $R''_{\{\ttt{I}, \ttt{J}\}}\}$, where all the relation names follow those in Section~\ref{sec&intro-ours}.
\done

We will refer to the conversion from $\Q'(\bm{\eta})$ to $\Q''(\bm{\eta})$ as {\em semi-join reduction} and call $\Q''(\bm{\eta})$ the {\em reduced query} under $\bm{\eta}$. It is rudimentary to verify: 
\myeqn{
    \join(\Q'(\bm{\eta})) = \join(\Q''(\bm{\eta})) = \cp(\Q''_\iso(\bm{\eta})) \times \join(\Q''_{\mit{light}}(\bm{\eta})). \label{eqn&framework-Q'=Q''} 
}

\subsection{The Isolated Cartesian Product Theorem} \label{sec&cpthm}

As shown in \eqref{eqn&Q''-iso}, $\Q''_\iso(\bm{\eta})$ contains $|\I|$ unary relations, one for each isolated attribute in $\I$. Hence,  $\cp(\Q''_\iso(\bm{\eta}))$ is the cartesian product of all those relations. The size of $\cp(\Q''_\iso(\bm{\eta}))$ has a crucial impact on the efficiency of our join strategy because, as shown in Lemma~\ref{lmm&cpalg}, the load for computing a cartesian product depends on the cartesian product's size. To prove that our strategy is efficient, we want to argue that 
\myeqn{
        \sum_{\bm{\eta} \in \config(\Q, \H)}
        \Big|\cp(\Q''_\iso(\bm{\eta}))\Big|
        \label{eqn&cpsum-weak}
}
is low, namely, the cartesian products of all the configurations $\bm{\eta} \in \config(\Q, \H)$ have a small size overall. 

\vgap

It is easy to place an upper bound of $\lambda^{|\H|} \cdot m^{|\I|}$ on \eqref{eqn&cpsum-weak}. As each relation (trivially) has size at most $m$, we have $|\cp(\Q''_\iso(\bm{\eta}))| \le m^{|\I|}$. Given that $\H$ has at most $\lambda^{|\H|}$ different configurations, \eqref{eqn&cpsum-weak} is at most $\lambda^{|\H|} \cdot m^{|\I|}$. Unfortunately, the bound is not enough to establish the claimed performance of our MPC algorithm (to be presented in Section~\ref{sec&alg}). For that purpose, we will need to prove a tighter upper bound on \eqref{eqn&cpsum-weak} --- this is where the isolated cartesian product theorem (described next) comes in.  

\vgap

Given an arbitrary fractional edge packing $W$ of the hypergraph $\G$, we define 
\myeqn{
    W_\I &=& \sum_{Y \in \I} \textrm{weight of $Y$ under $W$}. 
    \label{eqn&W_I}
}
Recall that the weight of a vertex $Y$ under $W$ is the sum of $W(e)$ for all the edges $e \in \E$ containing $Y$. 



\begin{thm}[{\em The isolated cartesian product theorem}] \label{thm&cpthm}
    Let $\Q$ be a simple binary query whose relations have a total size of $m$. Denote by $\G$ the hypergraph defined by $\Q$. Consider an arbitrary subset $\H \subseteq \attset(\Q)$, where $\attset(\Q)$ is the set of attributes in the relations of $\Q$. Let $\I$ be the set of isolated vertices defined in \eqref{eqn&I}. Take an arbitrary fractional edge packing $W$ of $\G$.    
    It holds that
\myeqn{
        \sum_{\bm{\eta} \in \config(\Q, \H)}
        \Big|\cp(\Q''_\iso(\bm{\eta}))\Big|
        &\le&
        \lambda^{|\H| - W_\I}
        \cdot 
        m^{|\I|}
        \label{eqn&cpsum}
    }
    where $\lambda$ is the heavy parameter (Section~\ref{sec&taxonomy}), $\config(\Q,\H)$ is the set of configurations of $\H$ (Section~\ref{sec&taxonomy}), $\Q''_\iso(\bm{\eta})$ is defined in \eqref{eqn&Q''-iso}, and $W_\I$ is defined in \eqref{eqn&W_I}. 
\end{thm}

Theorem~\ref{thm&cpthm} is in the strongest form when $W_\I$ is maximized. Later in Section~\ref{sec&cpthm-corollary}, we will choose a specific $W$ that yields a bound sufficient for us to prove the efficiency claim on our join algorithm. 

\vgap

\extraspacing {\bf Proof of Theorem~\ref{thm&cpthm}.} We will construct a set $\Q^*$ of relations such that $\join(\Q^*)$ has a result size at least the left-hand side of \eqref{eqn&cpsum}. Then, we will prove that the hypergraph of $\Q^*$ has a fractional edge covering that (by the AGM bound; Lemma~\ref{lmm&agm}) implies an upper bound on $|\join(\Q^*)|$ matching the right-hand side of \eqref{eqn&cpsum}.

\vgap 

\begin{figure}
    \includegraphics[height=35mm]{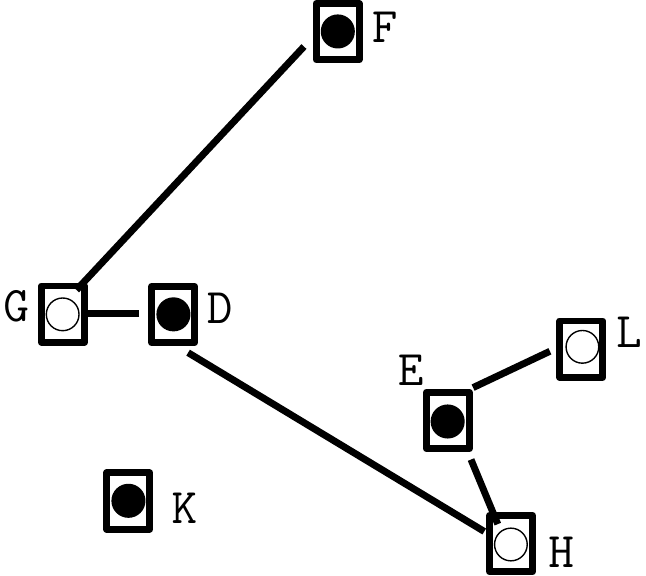} 
    \figcapup 
    \caption{Illustration of $\Q^*$} 
    \label{fig&Qstar}
\end{figure}    

Initially, set $\Q^*$ to $\emptyset$. For every cross edge $e \in \E$ incident to a vertex in $\I$, add to $\Q^*$ a relation $R^*_e = R_e$. For every $X \in \H$, add a unary relation $R^*_{\{X\}}$ to $\Q^*$ which consists of all the heavy values on $X$; note that $R^*_{\{X\}}$ has at most $\lambda$ tuples. Finally, for every $Y \in \I$, add a unary relation $R^*_{\{Y\}}$ to $\Q^*$ which  contains all the heavy and light values on $Y$. 

\vgap

Define $\G^* = (\V^*, \E^*)$ as the hypergraph defined by $\Q^*$. Note that $\V^* = \I \cup \H$, while $\E^*$ consists of all the cross edges in $\G$ incident to a vertex in $\I$, $|\H|$ unary edges $\{X\}$ for every $X \in \H$, and $|\I|$ unary edges $\{Y\}$ for every $Y \in \I$. 

 \extraspacing {\bf {\em Example (cont.).}} Figure~\ref{fig&Qstar} shows the hypergraph of the $\Q^*$ constructed. As before, a box and a segment represent a unary and a binary edge, respectively. Recall that $\H = \{\ttt{D}, \ttt{E}, \ttt{F}, \ttt{K}\}$ and $\I = \{\ttt{G}, \ttt{H}, \ttt{L}\}$.  
\done

\begin{lem} 
    $
    \sum_{\bm{\eta'} \in \config(\Q, \H)}
    \Big|\cp(\Q''_\iso(\bm{\eta'}))\Big|
    \le
    |\join(\Q^*)|. 
    $ 
\end{lem}
\begin{proof}
    
We will prove 
    \myeqn{
        \bigcup_{\bm{\eta'} \in \config(\Q, \H)}
        \cp(\Q''_\iso(\bm{\eta'})) \times \{\bm{\eta'}\} 
        &\subseteq&
        \join(\Q^*). \label{lmm&cpthm-cp-size-<=-constructed-join}
    } 
    from which the lemma follows. 
    
    \vgap

    Take a tuple $\bm{u}$ from the left-hand side of \eqref{lmm&cpthm-cp-size-<=-constructed-join}, and set $\bm{\eta'} = \bm{u}[\H]$. Based on the definition of $\Q''_\iso(\bm{\eta'})$, it is easy to verify that $\bm{u}[e] \in R_e$ for every cross edge $e \in \E$ incident a vertex in $\I$; hence, $\bm{u}[e] \in R^*_e$. Furthermore, $\bm{u}(X) \in R^*_{\{X\}}$ for every $X \in \H$ because $\bm{u}(X) = \bm{\eta'}(X)$ is a heavy value. Finally, obviously $\bm{u}(Y) \in R^*_{\{Y\}}$ for every $Y \in \I$. All these facts together ensure that $\bm{u} \in \join(\Q^*)$.
\end{proof}

\vgap

\begin{lem} \label{lmm&cpthm-tightcovering}
    $\G^*$ admits a tight fractional edge covering $W^*$ satisfying
    $\sum_{X \in \H} W^*(\{X\}) = |\H|-W_\I$.
\end{lem}

\begin{proof}
    We will construct a desired function $W^*$ from the fractional edge packing $W$ in Theorem~\ref{thm&cpthm}. 
    
    \vgap
    
    For every cross edge $e \in \E$ incident to a vertex in $\I$, set $W^*(e) = W(e)$. Every edge in $\E$ incident to $Y \in \I$ must be a cross edge. Hence, $\sum_{\textrm{binary $e$} \in \E^*: Y \in e} W^*(e)$ is precisely the weight of $Y$ under $W$.

    \vgap 
    
    Next, we will ensure that each attribute $Y \in \I$ has a weight 1 under $W^*$. Since $W$ is a fractional edge packing of $\G$, it must hold that $\sum_{\textrm{binary $e$} \in \E^*: Y \in e} W(e) \le 1$. This permits us to assign the following weight to the unary edge $\{Y\}$: 
    \myeqn{
        W^*(\{Y\})
        &=&
        1 - \sum_{\textrm{binary $e$} \in \E^*: Y \in e} W(e). \nn 
    }
    
    Finally, in a similar way, we make sure that each attribute $X \in \H$ has a weight 1 under $W^*$ by assigning: 
    \myeqn{
        W^*(\{X\})
        &=&
        1 - \sum_{\textrm{binary $e$} \in \E^*: X \in e} W(e). \nn 
    }
    This finishes the design of $W^*$, which is now a tight fractional edge covering of $\G^*$. 
    
    \vgap
    
    Clearly:
    \myeqn{ 
        \sum_{X \in \H} W^*(\{X\})
        &=&
        |\H| - \sum_{X \in \H} \left( \sum_{\textrm{binary $e$} \in \E^*: X \in e} W(e) \right). \label{eqn&cpsum-strong-tightcovering-help}
    }
    Every binary edge $e \in \E^*$ contains a vertex in $\H$ and a vertex in $\I$. Therefore: 
    \myeqn{
        \sum_{X \in \H} \left(\sum_{\textrm{binary $e$} \in \E^*: X \in e} W(e) \right)
        &=&
        \sum_{Y \in \I} \left( \sum_{\textrm{binary $e$} \in \E^*: Y \in e} W(e) \right)
        = 
        W_\I. \nn
    }
    Putting together the above equation with \eqref{eqn&cpsum-strong-tightcovering-help} completes the proof.
\end{proof}


The AGM bound in Lemma~\ref{lmm&agm} tells us that 
\myeqn{
    \join(\Q^*)
    &\le&
    \prod_{e \in \E^*} |R^*_e|^{W^*(e)} 
    \nn \\ 
    &
    =&  
    \Big(
    \prod_{X \in \H} |R^*_{\{X\}}|^{W^*(\{X\})} 
    \Big)
    \Big(\prod_{Y \in \I} 
    \,\,
    \prod_{\textrm{$e$} \in \E^* : Y \in e} 
    |R^*_e|^{W^*(e)} 
    \Big)
    \nn
    \\
    &\le&  
    \Big(
    \prod_{X \in \H} \lambda^{W^*(\{X\})}  
    \Big)
    \Big(\prod_{Y \in \I} 
    \,\,
    \prod_{e \in \E^* : Y \in e} 
    m^{W^*(e)} 
    \Big) \nn \\
    &&
    \explain{applying $|R^*_{\set{X}}| \le \lambda$ and $|R^*_e| \le m$} \nn \\
    &\le& 
    \lambda^{|\H|-W_\I} \cdot m^{|\I|} \nn \\ 
    &&
    \textrm{(by Lemma~\ref{lmm&cpthm-tightcovering} and $\sum_{e \in \E^*: Y \in e} W^*(e) = 1$ for each $Y$ due to tightness of $W^*$)} \nn
}
which completes the proof of Theorem~\ref{thm&cpthm}. \qed

\subsection{A Subset Extension of Theorem~\ref{thm&cpthm}} \label{sec&cpthm2}

Remember that $\Q''_\iso(\bm{\eta})$ contains a relation $R''_X(\bm{\eta})$ (defined in \eqref{eqn&R''-X}) for every attribute $X \in \I$. Given a non-empty subset $\J \subseteq \I$, define 
\myeqn{
    \Q''_\J (\bm{\eta})
    &=&
    \{
    R''_X(\bm{\eta}) \bigm| X \in \J 
    \}. \label{eqn&Q''_J} 
}
Note that $\join(\Q''_\J(\bm{\eta}))$ is the cartesian product of the relations in $\Q''_\J(\bm{\eta})$.

\vgap

Take an arbitrary fractional edge packing $W$ of the hypergraph $\G$. Define 
\myeqn{
    W_\J &=& \sum_{Y \in \J} \textrm{weight of $Y$ under $W$}. 
    \label{eqn&W_J}
}

\vgap

We now present a general version of the isolated cartesian product theorem:

\begin{thm} \label{thm&cpthm2}
    Let $\Q$ be a simple binary query whose relations have a total size of $m$. Denote by $\G$ the hypergraph defined by $\Q$. Consider an arbitrary subset $\H \subseteq \attset(\Q)$, where $\attset(\Q)$ is the set of attributes in the relations of $\Q$. Let $\I$ be the set of isolated vertices defined in \eqref{eqn&I} and $\J$ be any non-empty subset of $\I$.      
    Take an arbitrary fractional edge packing $W$ of $\G$. It holds that
    \myeqn{
        \sum_{\bm{\eta} \in \config(\Q, \H)}
        \Big|\cp(\Q''_\J(\bm{\eta}))\Big|
        &\le&
        \lambda^{|\H| - W_\J} \cdot m^{|\J|}.
        \label{eqn&cpsum2}
    }
    where $\lambda$ is the heavy parameter (see Section~\ref{sec&taxonomy}), $\config(\Q,\H)$ is the set of configurations of $\H$ (Section~\ref{sec&taxonomy}), $\Q''_\J$ is defined in \eqref{eqn&Q''_J}, and $W_\J$ is defined in \eqref{eqn&W_J}.
\end{thm}

\begin{proof}
    We will prove the theorem by reducing it to Theorem~\ref{thm&cpthm}. Define $\overline{\J} = \I \setminus \J$ and 
    \myeqn{
        \tilde{\Q}
        &=&
        \{R \in Q \mid \scheme(R) \cap \overline{\J} = \emptyset\}. \nn
    }
    One can construct $\tilde{\Q}$ alternatively as follows. First, discard from $\Q$ every relation whose scheme contains an attribute in $\overline{\J}$. Then, $\tilde{\Q}$ consists of the relations remaining in $\Q$.

    \vgap
    
    Denote by $\tilde{\G} = (\tilde{\V}, \tilde{\E})$ the hypergraph defined by $\tilde{\Q}$.    
    Set $\tilde{\H} = \H \cap \attset(\tilde{\Q})$ and $\tilde{\L}  = \attset(\tilde{\Q}) \setminus \tilde{\H}$. $\J$ is precisely the set of isolated attributes decided by $\tilde{\Q}$ and $\tilde{\H}$.\footnote{Let $\tilde{\I}$ be the set of isolated attributes after removing $\tilde{\H}$ from $\tilde{\G}$. We want to prove $\J = \tilde{\I}$. It is easy to show $\J \subseteq \tilde{\I}$. To prove $\tilde{\I} \subseteq \J$, suppose that there is an attribute $X$ such that $X \in \tilde{\I}$ but $X \notin \J$. As $X$ appears in $\tilde{\G}$, we know $X \notin \I$. Hence, $\G$ must contain an edge $\set{X,Y}$ with $Y \notin \H$. This means $Y \notin \I$, because of which the edge $\set{X,Y}$ is disjoint with $\overline{\J}$ and thus must belong to $\tilde{\G}$. But this contradicts the fact $X \in \tilde{\I}$.} 
    
    \vgap
    
    Define a function $\tilde{W}: \tilde{\E} \rightarrow [0, 1]$ by setting $\tilde{W}(e) = W(e)$ for every $e \in \tilde{\E}$. $\tilde{W}$ is a fractional edge packing of $\tilde{\G}$.    
    Because every edge $e \in \E$ containing an attribute in $\J$ is preserved in $\tilde{\E}$,\footnote{Suppose that there is an edge $e = \{X, Y\}$ such that $X \in \J$ and yet $e \notin \tilde{\E}$. It means that $Y \in \bar{\J} \subseteq \I$. But then $e$ is incident on two attributes in $\I$, which is impossible.} we have $W_\J = {\tilde{W}}_\J$.   
    Applying Theorem~\ref{thm&cpthm} to $\tilde{\Q}$ gives: 
    \myeqn{
        \sum_{\tilde{\bm{\eta}} \in \config(\tilde{\Q}, \tilde{\H})}
        \Big|\cp(\tilde{\Q}''_\iso(\tilde{\bm{\eta}}))\Big|
        \,\le\,
        \lambda^{|\tilde{\H}| - {\tilde{W}}_\J} \cdot m^{|\J|} 
        \,=\,
        \lambda^{|\tilde{\H}| - W_\J} \cdot m^{|\J|}
        .
        \label{eqn&cpthm2-help1}
    }
    It remains to show 
    \myeqn{
        \sum_{\bm{\eta} \in \config(\Q, \H)}
        \Big|\cp(\Q''_\J(\bm{\eta}))\Big|
        &\le&
        \lambda^{|\H| - |\tilde{\H}|}
        \sum_{\tilde{\bm{\eta}} \in \config(\tilde{\Q}, \tilde{\H})}
        \Big|\cp(\tilde{\Q}''_\iso(\tilde{\bm{\eta}}))\Big|
        \label{eqn&cpthm2-help2}
    }
    after which Theorem~\ref{thm&cpthm2} will follow from \eqref{eqn&cpthm2-help1} and \eqref{eqn&cpthm2-help2}.
    
    \vgap 
    
    For each configuration $\bm{\eta} \in \config(\Q, \H)$, we can find $\tilde{\bm{\eta}} = \bm{\eta}[\tilde{\H}] \in \config(\tilde{Q}, \tilde{H})$ such that $\cp(\Q''_\J(\bm{\eta})) = \cp(\tilde{\Q}''_\iso(\tilde{\bm{\eta}}))$. The correctness of \eqref{eqn&cpthm2-help2} follows from the fact that at most $\lambda^{|\H| - |\tilde{\H}|}$ configurations $\bm{\eta} \in \config(\Q, \H)$ correspond to the same $\tilde{\bm{\eta}}$. 
\end{proof}

\subsection{A Weaker Result} \label{sec&cpthm-corollary} 

One issue in applying Theorem~\ref{thm&cpthm2} is that the quantity $|\H| - W_\J$ is not directly related to the fractional edge covering number $\rho$ of $\Q$. The next lemma gives a weaker result that addresses the issue to an extent sufficient for our purposes in Section~\ref{sec&alg}: 

\begin{lem} \label{lmm&cpthm2-weakver}
    Let $\Q$ be a simple binary query who relations have a total size of $m$. Denote by $\G$ the hypergraph defined by $\Q$. Consider an arbitrary subset $\H \subseteq \attset(\Q)$, where $\attset(\Q)$ is the set of attributes in the relations of $\Q$. Define $\L = \attset(\Q) \setminus \H$ and $\I$ as the set of isolated vertices in $\L$ (see \eqref{eqn&I}). 
    For any non-empty subset $\J \subseteq \I$, it holds that
    \myeqn{
        \sum_{\bm{\eta} \in \config(\Q,\H)}
        \Big|\cp(\Q''_\J(\bm{\eta}))\Big|
        &\le&
        \lambda^{2 \rho - |\J| - |\L|}
        \cdot 
        m^{|\J|}
        \label{eqn&cpthm2-weakver}
    }
    where $\rho$ is the fractional edge covering number of $\G$, $\lambda$ is the heavy parameter (Section~\ref{sec&taxonomy}), $\config(\Q,\H)$ is the set of configurations of $\H$ (Section~\ref{sec&taxonomy}), and $\Q''_\J(\bm{\eta})$ is defined in \eqref{eqn&Q''_J}.
\end{lem}

\begin{proof}
    Let $W$ be an arbitrary fractional edge packing of $\G$ satisfying the second bullet of Lemma~\ref{lmm&coverpack}. Specifically, the weight of $W$ is the fractional edge packing number $\tau$ of $\G$; and the weight of every vertex in $\G$ is either 0 or 1. Denote by $Z$ the set of vertices in $\G$ whose weights under $W$ are 0. Lemma~\ref{lmm&coverpack} tells us $\tau + \rho = |\attset(\Q)|$ and $\rho - \tau = |\Z|$. Set $\J_0 = \J \cap Z$ and $\J_1 = \J \setminus \J_0$.    
    Because $\J_0 \subseteq \Z$, we can derive: 
\myeqn{
    \tau + |\J_0| 
    &\le& 
    \rho \,\, \Rightarrow \nn \\
    |\attset(\Q)| - \rho + |\J_0| 
    &\le& 
    \rho \,\, \Rightarrow \nn \\ 
    (|\H| + |\L|) + (|\J| - |\J_1|) 
    &\le&
    2\rho \,\, \Rightarrow \nn \\
    |\H| - |\J_1| 
    &\le& 
    2\rho - |\J| - |\L|. \nn 
}
    Lemma~\ref{lmm&cpthm2-weakver} now follows from Theorem~\ref{thm&cpthm2} due to $|\J_1| = W_\J$, which holds because every vertex in $\J_1$ has weight 1 under $W$. 
\end{proof}

\noindent 
{\bf Remark.} The above lemma was the ``isolated cartesian product theorem'' presented in the preliminary version \cite{t20} of this work. The new version (i.e., Theorem~\ref{thm&cpthm2}) is more powerful and better captures the mathematical structure underneath. 

\section{An MPC Join Algorithm} \label{sec&alg}

This section will describe how to answer a simple binary join $\Q$ in the MPC model with load $\tO(m/p^{1/\rho})$. 

\vgap

We define a {\em statistical record} as a tuple $(R, X, x, \mathit{cnt})$, where $R$ is a relation in $\Q$, $X$ an attribute in $\scheme(R)$, $x$ a value in $\dom$, and $\mathit{cnt}$ the number of tuples $\bm{u} \in R$ with $\bm{u}(X) = x$. Specially, $(R, \emptyset, \mathit{nil}, \mathit{cnt})$ is also a statistical record where $\mathit{cnt}$ gives the number of tuples in $R$ that use only light values. A {\em histogram} is defined as the set of statistical records for all possible $R$, $X$, and $x$ satisfying (i) $\mathit{cnt} = \Omega(m/p^{1/\rho})$ or (ii) $X = \emptyset$ (and, hence $x = \mathit{nil}$); note that there are only $O(p^{1/\rho})$ such records. We assume that every machine has a local copy of the histogram. By resorting to standard MPC sorting algorithms \cite{g99,hyt19}, the assumption can be satisfied with a preprocessing that takes constant rounds and load $\tO(p^{1/\rho} + m/p)$. 

\vgap 

Henceforth, we will fix the heavy parameter
\myeqn{
    \lambda &=& \Theta(p^{1/(2\rho)}) \nn
}
and focus on explaining how to compute \eqref{eqn&framework-goal} for an arbitrary subset $\H$ of $\attset(\Q)$.
As $\attset(\Q)$ has $2^k = O(1)$ subsets (where $k = |\attset(\Q)|$), processing them all in parallel increases the load by only a constant factor and, as guaranteed by \eqref{eqn&taxonomy-goal}, discovers the entire $\join(\Q)$. 

\vgap

Our algorithm produces \eqref{eqn&framework-goal} in three steps: 
\begin{enumerate}
    \item Generate the input relations of the residual query $\Q'(\bm{\eta})$ of every configuration $\bm{\eta}$ of $\H$ (Section~\ref{sec&framework-removeH}). 
    \item Generate the input relations of the reduced query $\Q''(\bm{\eta})$ of every $\bm{\eta}$ (Section~\ref{sec&framework-semi}). 
    \item Evaluate $\Q''(\bm{\eta})$ for every $\bm{\eta}$.
\end{enumerate}
The number of configurations of $\H$ is $O(\lambda^{|\H|}) = O(\lambda^k) = O(p^{k/(2\rho)})$, which is $O(p)$ because $\rho \ge k/2$ by 
the first bullet of Lemma~\ref{lmm&coverpack}. Next, we elaborate on the details of each step. 

 \extraspacing {\bf Step 1.} Lemma~\ref{lmm&taxonomy-inputsize} tells us that the input relations of all the residual queries have at most $m \cdot \lambda^{k-2}$ tuples in total. We  allocate $p'_{\bm{\eta}} = \lceil p \cdot \fr{m_{\bm{\eta}}}{\Theta(m \cdot \lambda^{k-2})} \rceil$ machines to store the relations of $\Q'(\bm{\eta})$, making sure that $\sum_{\bm{\eta}} p'_{\bm{\eta}} \le p$. Each machine keeps on average $$O(m_\bm{\eta}/p'_\bm{\eta}) = O(m \cdot \lambda^{k-2} / p) = O(m/p^{1/\rho})$$ tuples, where the last equality used $\rho \ge k/2$. Each machine $i \in [1, p]$ can use the histogram to calculate the input size $m_{\bm{\eta}}$ of $\Q'(\bm{\eta})$ precisely for each $\bm{\eta}$; it can compute locally the id range  of the $m_{\bm{\eta}}$ machines responsible for $Q'(\bm{\eta})$. If a tuple $\bm{u}$ in the local storage of machine $i$ belongs to $Q'(\bm{\eta})$, the machine sends $\bm{u}$ to a random machine within that id range. Standard analysis shows that each of the $m_{\bm{\eta}}$ machines receives asymptotically the same number of tuples of $\Q'(\bm{\eta})$ (up to an $\tO(1)$ factor) with probability at least $1 - 1/p^c$ for an arbitrarily large constant $c$. Hence, Step 1 can be done in a single round with load $\tO(m/p^{1/\rho})$ with probability at least $1 - 1/p^c$.

 \extraspacing {\bf Step 2.} Now that all the input relations of each $\Q'(\bm{\eta})$ have been stored on $p'_{\bm{\eta}}$ machines, the semi-join reduction in Section~\ref{sec&framework-semi} that converts  $\Q'(\bm{\eta})$ to  $\Q''(\bm{\eta})$ is a standard process that can be accomplished \cite{hyt19} with sorting in $O(1)$ rounds entailing a load of $\tO(m_{\bm{\eta}}/p'_{\bm{\eta}}) = \tO(m/p^{1/\rho})$; see also \cite{kbs16} for a randomized algorithm that performs fewer rounds.

 \extraspacing {\bf Step 3.} This step starts by letting each machine know about the value of $|\cp(\Q''_\iso(\bm{\eta}))|$ for every $\bm{\eta}$. For this purpose, each machine broadcasts to all other machines how many tuples it has in $R''_X(\bm{\eta})$ for every $X \in \I$ and every $\bm{\eta}$. Since there are $O(p)$ different $\bm{\eta}$, $O(p)$ numbers are sent by each machine, such that the load of this round is $O(p^2)$. From the numbers received, each machine can independently figure out the values of all $|\cp(\Q''_\iso(\bm{\eta}))|$. 
 

\vgap 

We allocate 
\myeqn{
    p''_{\bm{\eta}}
    &=&
     \Theta\left(\lambda^{|\L|} + p \cdot \sum_{\textrm{non-empty $\J \subseteq \I$}} \fr{|\cp(\Q''_\J(\bm{\eta}))|}{\lambda^{2\rho - |\J| - |\L|}
        \cdot 
        m^{|\J|}
        }
    \right)
        \label{eqn&alg-machine-num}
}
machines for computing $\Q''(\bm{\eta})$. Notice that  
\myeqn{
    \sum_{\bm{\eta}} p''_\bm{\eta} = 
    O
    \left(
    \sum_{\bm{\eta}}
        \lambda^{|\L|}
    \right) 
    +
    O
    \left(p \cdot 
    \sum_{\textrm{non-empty $\J \subseteq \I$}}
    \sum_{\bm{\eta}} \fr{|\cp(\Q''_\J(\bm{\eta}))|}{\lambda^{2\rho - |\J| - |\L|}
        \cdot 
        m^{|\J|}} \right) 
    = O(p)
     \nn 
}
where the equality used Lemma~\ref{lmm&cpthm2-weakver}, the fact that $\I$ has constant non-empty subsets, and that $\sum_{\bm{\eta}} \lambda^{|\L|} \le \lambda^{|\H|} \cdot \lambda^{|\L|} = \lambda^k \le p$. We can therefore adjust the constants in \eqref{eqn&alg-machine-num} to make sure that the total number of machines needed by all the configurations is at most $p$. 

\begin{lem} 
    $\Q''(\bm{\eta})$ can be answered in one round with load $O(m/p^{1/\rho})$ using $p''_\bm{\eta}$ machines, subject to a failure probability of at most $1/p^c$ where $c$ can be set to an arbitrarily large constant.
\end{lem}

\begin{proof}
    As shown in \eqref{eqn&framework-Q'=Q''}, $\join(\Q''(\bm{\eta}))$ is the cartesian product of $\cp(\Q''_\iso(\bm{\eta}))$ and $\join(\Q''_\mit{light}(\bm{\eta}))$. We deploy $\Theta(p''_{\bm{\eta}} / \lambda^{|\L| - |\I|})$ machines to compute $\cp(\Q''_\iso(\bm{\eta}))$ in one round. 
    By Lemma~\ref{lmm&cpalg}, the load is 
    \myeqn{
        \tO\left( 
            \fr{|\cp(\Q''_\J(\bm{\eta}))|^{1/|\J|}}
            {\left(\fr{p''_{\bm{\eta}}}{\lambda^{{|\L| - |\I|}}} \right)^{1/|\J|}}
        \right) \label{eqn&alg-oneinst-machine-num}
    }
    for some non-empty $\J \subseteq \I$. \eqref{eqn&alg-machine-num} guarantees that 
    \myeqn{
    p''_{\bm{\eta}}
    &=&
     \Omega\left(p \cdot \fr{|\cp(\Q''_\J(\bm{\eta}))|}{\lambda^{2\rho - |\J| - |\L|}
        \cdot 
        m^{|\J|}
        }
    \right)
    \nn
    }
    with which we can derive
    \myeqn{
        \eqref{eqn&alg-oneinst-machine-num}
        =
        \tO\left(
        \fr{m \cdot \lambda^{\fr{2 \rho - |\J| - |\I|}{|\J|} }}
        {p^{1/|\J|}}
        \right)
        =
         \tO\left(
        \fr{m \cdot \lambda^{\fr{2 \rho - 2|\J|}{|\J|} }}
        {p^{1/|\J|}}
        \right)
        =
        \tO\left(
        \fr{m \cdot p^{\fr{2\rho - 2|\J|}{2 \rho |\J|} }}
        {p^{1/|\J|}}
        \right)
        =
        \tO\left(
        \fr{m}
        {p^{1/\rho}}
        \right).
        \nn
    }

    Regarding $\Q''_\mit{light}(\bm{\eta})$, first note that $\attset(\Q''_\mit{light}(\bm{\eta})) = \L \setminus \I$. If $\L \setminus \I$ is empty, no $\Q''_\mit{light}(\bm{\eta})$ exists and $\join(\Q''(\bm{\eta})) = \cp(\Q''_\iso(\bm{\eta}))$. The subsequent discussion considers that $\L \setminus \I$ is not empty.         
    As the input relations of $\Q''_\mit{light}(\bm{\eta})$ contain only light values, $\Q''_\mit{light}(\bm{\eta})$ is skew-free if a share of $\lambda$ is assigned to each attribute in $\L \setminus \I$. By Lemma~\ref{lmm&skewfree}, $\join(\Q''_\mit{light}(\bm{\eta}))$ can be computed in one round with load $\tO(m/\lambda^2) = \tO(m/p^{1/\rho})$ using $\Theta(\lambda^{|\L \setminus \I|})$ machines, subject to a certain failure probability $\delta$. As $\lambda^{|\L \setminus \I|} \ge \lambda$ which is a polynomial of $p$, Lemma~\ref{lmm&skewfree} allows us to make sure $\delta \le 1/p^c$ for any constant $c$. 
    
    \vgap 
    
    By combining the above discussion with Lemma~\ref{lmm&cpcomp}, we conclude that $\join(\Q''(\bm{\eta}))$ can be computed in one round with load $\tO(m/p^{1/\rho})$ using $p''_{\bm{\eta}}$ machines, subject to a failure probability at most $\delta \le 1/p^c$.  
\end{proof}

Overall, the load of our algorithm is $\tO(p^{1/\rho} + p^2 + m/p^{1/\rho})$. This brings us to our second main result: 

\begin{thm} \label{thm&alg-main}
    Given a simple binary join query with input size $m \ge p^3$ and a fractional edge covering number $\rho$, we can answer it in the MPC model using $p$ machines in constant rounds with load $\tO(m/p^{1/\rho})$, subject to a failure probability of at most $1/p^c$ where $c$ can be set to an arbitrarily large constant.
\end{thm}

\section{Concluding Remarks} \label{sec&conclusion} 

This paper has introduced an algorithm for computing a natural join over binary relations under the MPC model. Our algorithm performs a constant number of rounds and incurs a load of $\tO(m/p^{1/\rho})$ where $m$ is the total size of the input relations, $p$ is the number of machines, and $\rho$ is the fractional edge covering number of the query. The load matches a known lower bound up to a polylogarithmic factor. Our techniques heavily rely on a new finding, which we refer to as the {\em isolated cartesian product theorem}, on the join problem's mathematical structure.

\vgap 

We conclude the paper with two remarks: 

\begin{itemize}
    \item The assumption $p^3 \le m$ can be relaxed to $p \le m^{1-\eps}$ for an arbitrarily small constant $\eps> 0$. Recall that our algorithm incurs a load of $\tO(p^{1/\rho} + p^2 + m/p^{1/\rho})$ where the terms $\tO(p^{1/\rho})$ and $\tO(p^2)$ are both due to the computation of statistics (in preprocessing and Step 2, respectively). In turn, these statistics are needed to allocate machines for subproblems. By using the machine-allocation techniques in \cite{hyt19}, we can avoid most of the statistics communication and reduce the load to $\tO(p^\eps + m/p^{1/\rho})$. 
    

    \item In the external memory (EM) model \cite{av88}, we have a machine equipped with $M$ words of internal memory and an unbounded disk that has been formatted into {\em blocks} of size $B$ words. An {I/O} either reads a block of $B$ words from the disk to the memory, or overwrites a block with $B$ words in the memory. A join query $\Q$ is considered solved if every tuple $\bm{u} \in \Q$ has been generated in memory at least once. The challenge is to design an algorithm to achieve the purpose with as few I/Os as possible. There exists a reduction \cite{kbs16} that can be used to convert an MPC algorithm to an EM counterpart. Applying the reduction on our algorithm gives an EM algorithm that solves $\Q$ with $\tO(\fr{m^\rho}{B \cdot M^{\rho-1}})$ I/Os, provided that $M \ge m^c$ for some positive constant $c < 1$ that depends on $\Q$. The I/O complexity can be shown to be optimal up to a polylogarithmic factor using the lower-bound arguments in \cite{hqt16,ps14}. We suspect that the constraint $M \ge m^c$ can be removed by adapting the isolated cartesian product theorem to the EM model. 
\end{itemize}


\bibliographystyle{alphaurl}
\bibliography{ref}

\begin{thebibliography}{10}

\bibitem{ahv95}
Serge Abiteboul, Richard Hull, and Victor Vianu.
\newblock {\em Foundations of Databases}.
\newblock Addison-Wesley, 1995.

\bibitem{aba+09}
Azza Abouzeid, Kamil Bajda-Pawlikowski, Daniel~J. Abadi, Alexander Rasin, and
  Avi Silberschatz.
\newblock Hadoopdb: An architectural hybrid of mapreduce and dbms technologies
  for analytical workloads.
\newblock {\em Proceedings of the VLDB Endowment ({PVLDB})}, 2(1):922--933,
  2009.

\bibitem{au11}
Foto~N. Afrati and Jeffrey~D. Ullman.
\newblock Optimizing multiway joins in a map-reduce environment.
\newblock {\em IEEE Transactions on Knowledge and Data Engineering ({TKDE})},
  23(9):1282--1298, 2011.

\bibitem{av88}
Alok Aggarwal and Jeffrey~Scott Vitter.
\newblock The input/output complexity of sorting and related problems.
\newblock {\em Communications of the ACM ({CACM})}, 31(9):1116--1127, 1988.

\bibitem{agm13}
Albert Atserias, Martin Grohe, and Daniel Marx.
\newblock Size bounds and query plans for relational joins.
\newblock {\em SIAM Journal on Computing}, 42(4):1737--1767, 2013.

\bibitem{bks17c}
Paul Beame, Paraschos Koutris, and Dan Suciu.
\newblock Communication steps for parallel query processing.
\newblock {\em Journal of the ACM ({JACM})}, 64(6):40:1--40:58, 2017.

\bibitem{c71}
Stephen~A. Cook.
\newblock The complexity of theorem-proving procedures.
\newblock In {\em Proceedings of {ACM} Symposium on Theory of Computing
  ({STOC})}, pages 151--158, 1971.

\bibitem{dg04}
Jeffrey Dean and Sanjay Ghemawat.
\newblock Mapreduce: Simplified data processing on large clusters.
\newblock In {\em Proceedings of USENIX Symposium on Operating Systems Design
  and Implementation ({OSDI})}, pages 137--150, 2004.

\bibitem{g99}
Michael~T. Goodrich.
\newblock Communication-efficient parallel sorting.
\newblock {\em SIAM Journal of Computing}, 29(2):416--432, 1999.

\bibitem{hyt19}
Xiao Hu, Ke~Yi, and Yufei Tao.
\newblock Output-optimal massively parallel algorithms for similarity joins.
\newblock {\em ACM Transactions on Database Systems ({TODS})}, 44(2):6:1--6:36,
  2019.

\bibitem{hqt16}
Xiaocheng Hu, Miao Qiao, and Yufei Tao.
\newblock {I/O}-efficient join dependency testing, loomis-whitney join, and
  triangle enumeration.
\newblock {\em Journal of Computer and System Sciences ({JCSS})},
  82(8):1300--1315, 2016.

\bibitem{ks17}
Bas Ketsman and Dan Suciu.
\newblock A worst-case optimal multi-round algorithm for parallel computation
  of conjunctive queries.
\newblock In {\em Proceedings of ACM Symposium on Principles of Database
  Systems ({PODS})}, pages 417--428, 2017.

\bibitem{kbs16}
Paraschos Koutris, Paul Beame, and Dan Suciu.
\newblock Worst-case optimal algorithms for parallel query processing.
\newblock In {\em Proceedings of International Conference on Database Theory
  ({ICDT})}, pages 8:1--8:18, 2016.

\bibitem{kss18}
Paraschos Koutris, Semih Salihoglu, and Dan Suciu.
\newblock Algorithmic aspects of parallel data processing.
\newblock {\em Foundations and Trends in Databases}, 8(4):239--370, 2018.

\bibitem{n91}
Ilan Newman.
\newblock Private vs. common random bits in communication complexity.
\newblock {\em Information Processing Letters ({IPL})}, 39(2):67--71, 1991.

\bibitem{nprr18}
Hung~Q. Ngo, Ely Porat, Christopher Re, and Atri Rudra.
\newblock Worst-case optimal join algorithms.
\newblock {\em Journal of the ACM ({JACM})}, 65(3):16:1--16:40, 2018.

\bibitem{nrr13}
Hung~Q. Ngo, Christopher Re, and Atri Rudra.
\newblock Skew strikes back: new developments in the theory of join algorithms.
\newblock {\em {SIGMOD} Record}, 42(4):5--16, 2013.

\bibitem{ps14}
Rasmus Pagh and Francesco Silvestri.
\newblock The input/output complexity of triangle enumeration.
\newblock In {\em Proceedings of ACM Symposium on Principles of Database
  Systems ({PODS})}, pages 224--233, 2014.

\bibitem{su97}
Edward~R. Scheinerman and Daniel~H. Ullman.
\newblock {\em Fractional Graph Theory: A Rational Approach to the Theory of
  Graphs}.
\newblock Wiley, New York, 1997.

\bibitem{t20}
Yufei Tao.
\newblock A simple parallel algorithm for natural joins on binary relations.
\newblock In {\em Proceedings of International Conference on Database Theory
  ({ICDT})}, pages 25:1--25:18, 2020.

\bibitem{v14}
Todd~L. Veldhuizen.
\newblock Triejoin: {A} simple, worst-case optimal join algorithm.
\newblock In {\em Proceedings of International Conference on Database Theory
  ({ICDT})}, pages 96--106, 2014.

\bibitem{y81}
Mihalis Yannakakis.
\newblock Algorithms for acyclic database schemes.
\newblock In {\em Proceedings of Very Large Data Bases ({VLDB})}, pages 82--94,
  1981.

\end{thebibliography}

\end{document}